\documentclass[journal, onecolumn]{IEEEtran}

\usepackage{graphicx}
\usepackage{subfigure}
\usepackage{color}
\usepackage{cite}
\usepackage{version}
\usepackage{fullpage,epsfig,graphicx,psfrag,color,subfig}
\usepackage{amsmath}
\usepackage{amssymb}

\usepackage{xspace}
\usepackage{bbm}
\usepackage{mathrsfs}

\newcommand{\Ac}{\mathcal{A}}
\newcommand{\Bc}{\mathcal{B}}

\newcommand{\Rc}{\mathcal{R}}

\newcommand{\Uc}{\mathcal{U}}

\newcommand{\Xc}{\mathcal{X}}
\newcommand{\Yc}{\mathcal{Y}}

\newcommand{\aep}{{\mathcal{T}_{\epsilon}^{(n)}}}

\newcommand{\Xh}{{\hat{X}}}
\newcommand{\Yh}{{\hat{Y}}}
\newcommand{\Zh}{{\hat{Z}}}

\newcommand{\lh}{{\hat{l}}}

\newcommand{\xh}{{\hat{x}}}

\newcommand{\Xt}{{\tilde{X}}}
\newcommand{\Yt}{{\tilde{Y}}}

\def\d{\delta}
\def\e{\epsilon}

\DeclareMathOperator\E{\sf E}
\let\P\relax
\DeclareMathOperator\P{\sf P}

\newcommand{\Bern}{\mathrm{Bern}}

\def\textiid{i.i.d.\@\xspace}
\newcommand\iid{\ifmmode\text{ i.i.d. } \else \textiid \fi}

\newtheorem{theorem}{Theorem}
\newtheorem{proposition}{Proposition}
\newtheorem{corollary}{Corollary}

\title{
Multi-Terminal Source Coding With Action Dependent Side Information
}

\author{
\authorblockN{Yeow-Khiang Chia, Himanshu Asnani and Tsachy Weissman \\}
\authorblockA{Department of Electrical Engineering, Stanford University \\
Email:  ykchia@stanford.edu, asnani@stanford.edu, tsachy@stanford.edu}
}
\begin{document}
\maketitle
\begin{abstract}
We consider multi-terminal source coding with a single encoder and multiple 
decoders where either the encoder or the decoders can take cost constrained
actions which affect the quality of the side information present at the
decoders. For the scenario where decoders take actions, we characterize the
rate-cost trade-off region for lossless source coding, and give an achievability
scheme for lossy source coding for two decoders which is optimum for a variety
of special cases of interest. For the case where the encoder takes actions, we
characterize the rate-cost trade-off for a class of lossless source coding
scenarios with multiple decoders. Finally, we also consider extensions to other multi-terminal source coding settings with actions, and characterize the rate -distortion-cost tradeoff for a case of successive refinement with actions.  
\end{abstract}
\section{Introduction} \label{sect:1}
The problem of source coding with decoder side information (S.I.) was introduced
in~\cite{Wyner}.  
S.I. acts as an important resource in rate distortion problems, where it can
significantly reduce the compression rate required.  In classical shannon theory
and in work building on ~\cite{Wyner}, S.I. is assumed to be either always
present or absent. However, in practical systems as we know, acquisition of S.I.
is costly,  the encoder or decoder has to expend resources to aquire side
information.  With this motivation,  the framework for the problem of source
coding with action-dependent side information (S.I.) was introduced
in~\cite{Permuter2010a}, where the authors considered the cases where the
encoder or decoder are allowed to take actions (with cost constraints) that
affect the quality or availability of the side information present at the
decoders, and in some settings, the encoder. As noted in~\cite{Permuter2010a},
one motivation for this setup is the case where the side information is obtained
via a sensor through a sequence of noisy measurements of the source sequence.
The sensor may have limited resources, such as acquisition time or power, in
obtaining the side information. This is therefore modeled by the cost constraint
on the action sequence to be taken at the decoder. Additional motivation for
considering this framework is given in~\cite{Permuter2010a}. We also refer
readers to recent work in ~\cite{AsnaniPermuterWeissman2010a},
\cite{Kittichokechai2010} for related Shannon theoretic scenarios invoking the
action framework.


In this paper, we extend the source coding with action framework to the 
case where there are multiple decoders, which can take actions that affect the
quality or availability of S.I. at each decoder, or where the encoder takes
actions that affect the quality or availability of S.I. at the decoders. As a
motivation for this framework, consider the following problem: An encoder
observes an i.i.d source sequence $X^n$ which it wishes to describe to two
decoders via a common rate limited link of rate $R$. The decoders, in addition
to observing the output of the common rate limited link, also have access to a
common sensor which gives side information $Y$ that is correlated with $X$.
However, because of contention or resource constraints, when decoder 1 observes
the side information, decoder 2 cannot access the side information and vice
versa. This problem is depicted in Figure~\ref{fig1}. Even in the absence of
cost constraints on the cost of switching to $1$ or $2$, this problem is interesting and
non-trivial. How should the decoders share the side information and what is the
optimum sequence of actions be conveyed and then taken by the decoder?
\begin{figure}
\begin{center}
\psfrag{y}{$Y_i$}
\psfrag{a1i}{$A_i(M)$}
\psfrag{a2i}{$A_i(M)$}
\psfrag{d1}[c]{Dec 1}
\psfrag{d2}[c]{Dec 2}
\psfrag{0}{$1$}
\psfrag{1}{$2$}
\psfrag{yh}{$X^n$}
\psfrag{xh}{$X^n$}
\psfrag{r}[c]{$M \in [1:2^{nR}]$}
\psfrag{Enc}[c]{Enc.}
\psfrag{X}[c]{$X^n$}
\includegraphics[width = 0.65\textwidth, height = 0.25\textwidth]{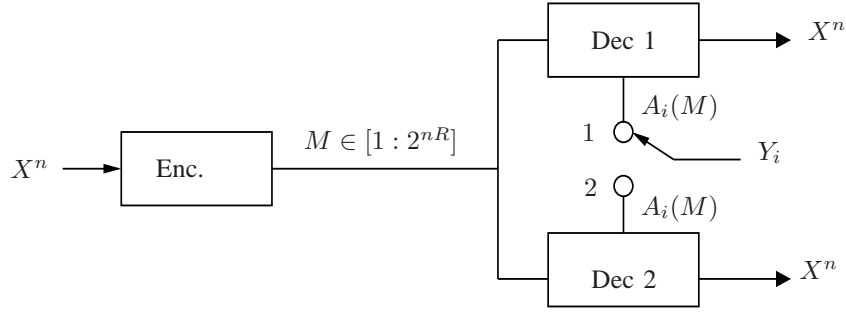} 
\end{center}
\caption{Lossless source coding with switching dependent side information. When the switch is at position $1$, decoder 1 observes the side information. When the switch is at position $2$, decoder 2 observes the side information.} \label{fig1}
\vspace{-18pt}
\end{figure} 

By posing the above problem in the framework of source coding with action
dependent side information, we solve it for the (near) lossless source coding
case, a special case of lossy source coding with switching dependent side
information, and give interpretations of the standard random binning and coding
arguments when specialized to this switching problem. As one example for the
implications of our findings, when $Y = X$,
we show that the optimum rate required for lossless source coding in the above
problem is $H(X)/2$ - clearly a lower bound on the required rate, but that it
suffices for perfect reconstruction of the source simultaneously at both
decoders is, at first glance, surprising. We devote a significant portion of
this paper to the setting where the side information at the decoders is obtained
through a switch  that determines which of the two decoders gets to observe the
side information, and obtain a complete characterization of the fundamental
performance limits in various scenarios involving such switching. The achieving
schemes in these scenarios are interesting in their own right, and also provide
insight into more general cases.

The rest of the paper is organized as follows. In section~\ref{sect:2}, we
provide 
formal definitions and problem formulations for the cases considered. In
section~\ref{sect:3}, we first consider the setting of lossless source coding
with decoders taking actions with cost constraints and give the optimum
rate-cost trade-off region for this setting. Next, we consider the setting of
lossy source coding decoders taking actions with cost constraints and give a
general achievability scheme for this setup. We then specialize our
achievability scheme to obtain the optimum rate-distortion and cost trade-off
region for a number of special cases. In section~\ref{sect:5}, we consider the
setting where actions are taken by the encoder. The rate-cost-distortion
tradeoff setting is open even for the single decoder case. Hence, we only
consider a special case of lossless source coding for which we can characterize
the rate-cost tradeoff. In section \ref{sect:6}, we extend our setup to two
other multiple users settings, including the case of successive refinement with
actions.  The paper is concluded in section \ref{sect:7}.
\section{Problem Definition} \label{sect:2}
In this section, we give formal definitions for, and focus on, the case 
where there are two decoders. Generalization of the definitions to $K$ decoders
is straightforward, and, as we indicate in subsequent sections, some of our
results hold in the $K$ decoders setting. We follow the notation
of~\cite{El-Gamal--Kim2010}. We use $A$ to denote the action random variable.
The distortion measure between sequences is defined in the usual way. Let $d:
\Xc \times \mathcal{\Xh} \rightarrow [0, \infty)$. Then, $d(x^n, \xh^n) :=
\frac{1}{n} \sum_{i=1}^n  d(x_i, \xh_i)$. The cost constraint is also defined in
the usual fashion: let $\Lambda(A^n):= \frac{1}{n}\sum_{i=1}^n \Lambda(A_i)$.
Throughout this paper, sources $(X^n, Y^n)$ are specified by the joint
distribution $p(x^n,y^n) = \prod_{i=1}^n p_{X,Y}(x_i,y_i)$ (i.i.d.). The
decoders obtain side information through a discrete memoryless \textit{action
channel} $\P_{Y_1, Y_2|X,A}$ specified by conditional distribution $p(y_1^n,
y_2^n|x^n, a^n) = \prod_{i=1}^n p_{Y_1, Y_2|X,A}(y_{1i}, y_{2i}|x_i, a_i)$, with
decoder $j$ obtaining side information $Y_j^n$ for $j\in\{1,2\}$. Extensions to
more than two sources or more than two channel outputs for multiple decoders are
straightforward. 
\subsection{Source coding with actions taken at the decoders}
This setting for two decoders is shown in figure~\ref{fig:2}. A $(n, 2^{nR})$ code for the above setting consists of one encoder
\begin{align*}
f &: \Xc^n \rightarrow M \in [1:2^{nR}],
\end{align*}
one \textit{joint action encoder at all decoders}
\begin{align*}
f_{A-Dec.} &: M \in [1:2^{nR}] \rightarrow \mathcal{A}^n,
\end{align*}
and two decoders
\begin{align*}
g_1 &: \Yc^n_1 \times [1:2^{nR}]\rightarrow \hat{\mathcal{X}}_1^n,\\
g_2 &: \Yc^n_2 \times [1:2^{nR}]\rightarrow \hat{\mathcal{X}}_2^n,
\end{align*}
Given a distortion-cost tuple $(D_1, D_2, C)$, a rate $R$ is said to be achievable if, for any $\e >0$ and $n$ sufficiently large, there exists $(n, 2^{nR})$ code such that{\allowdisplaybreaks
\begin{align*}
\E\left[ \frac{1}{n}\sum_{i=1}^n d_j(X_i, \Xh_{j,i})\right] &\le D_j + \e, \quad \mbox{j=1,2}, \\
\E\left[ \frac{1}{n}\sum_{i=1}^n \Lambda(A_i)\right] &\le C + \e.
\end{align*}}
The \textit{rate-distortion-cost region}, $\Rc(D_1, D_2, C)$, is defined as the infimum of all achievable rates.
\begin{figure}
\begin{center}
\psfrag{X}{$X^n$}
\psfrag{Enc}[c]{Enc.}
\psfrag{ya1}{$Y_1^n$}
\psfrag{ya2}{$Y_2^n$}
\psfrag{d1}[c]{Dec 1}
\psfrag{d2}[c]{Dec 2}
\psfrag{act}[c]{$P_{Y_1, Y_2|X,A}$}
\psfrag{yh}{$\Xh_1^n$}
\psfrag{xh}{$\Xh_2^n$}
\psfrag{A}[c]{$A^n(M)$}
\psfrag{r}[c]{$M \in [1:2^{nR}]$}
\includegraphics[width = 0.75\textwidth, height = 0.3\textwidth]{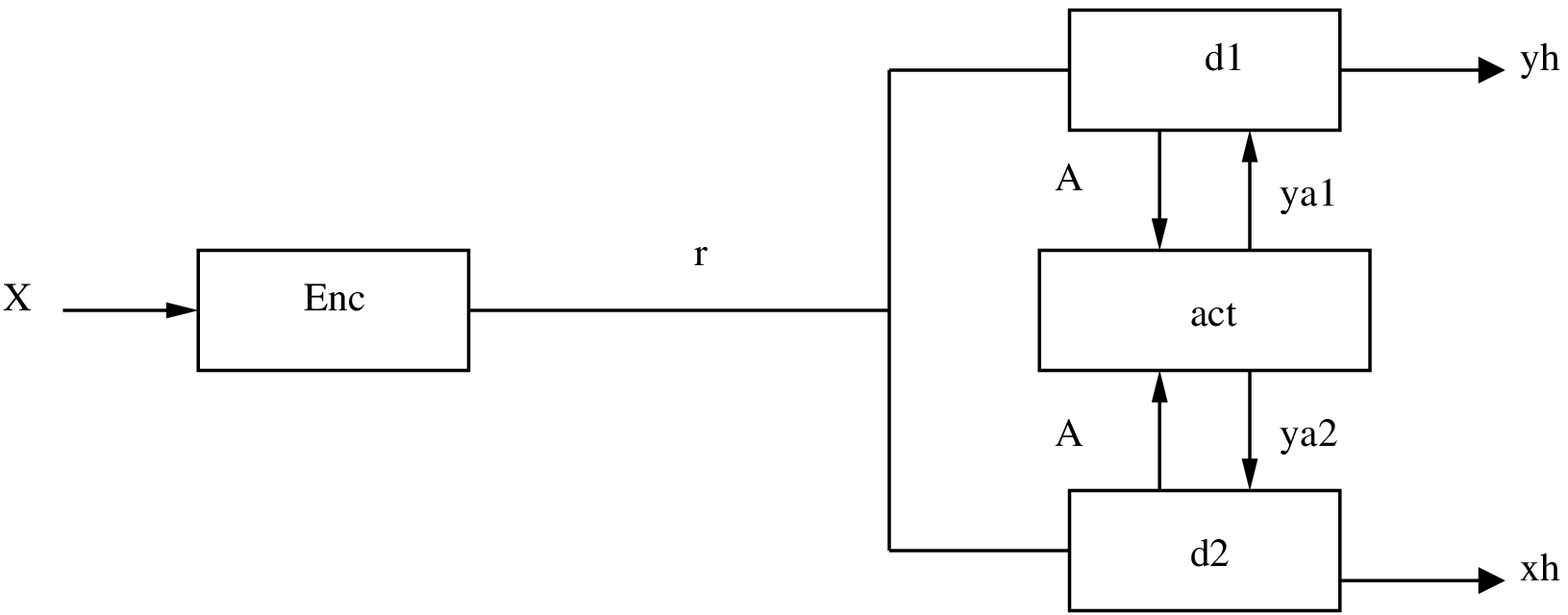} 
\end{center}
\caption{Lossy source coding with actions at the decoders.} \label{fig:2}
\vspace{-5pt}
\end{figure} 
\begin{figure}
\begin{center}
\psfrag{X}{$X^n$}
\psfrag{Enc}[c]{Enc.}
\psfrag{ya1}{$Y_1^n$}
\psfrag{ya2}{$Y_2^n$}
\psfrag{d1}[c]{Dec 1}
\psfrag{d2}[c]{Dec 2}
\psfrag{act}[c]{$P_{Y_1, Y_2|X,A}$}
\psfrag{yh}{$\Xh_1^n$}
\psfrag{xh}{$\Xh_2^n$}
\psfrag{Ax}[c]{$A^n(X^n)$}
\psfrag{r}[c]{$M \in [1:2^{nR}]$}
\includegraphics[width = 0.75\textwidth, height = 0.3\textwidth]{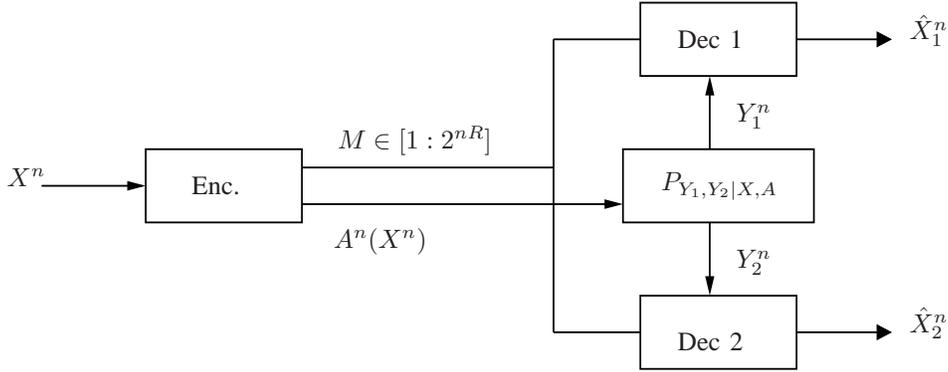} 
\end{center}
\caption{Lossy source coding with actions at the encoder.} \label{fig:3}
\vspace{-18pt}
\end{figure} 

\textit{Causal reconstruction with action dependent side information}: Some results in this paper involves the case of \textit{causal reconstruction}. In the case of causal reconstruction, the decoder reconstructs $\Xh_i$ based only on the received message $M$ and the side information up to time $i$. That is, 
\begin{align*}
g_{j,i} &: \Yc^i_j \times [1:2^{nR}]\rightarrow \hat{\mathcal{X}}_{j,i},
\end{align*} 
for $j \in \{1,2\}$ and $i \in [1:n]$.

\textit{Remark 2.1:} The case of the decoders taking separate actions $A_1$ and
$A_2$ respectively is a special case of our setup since we can write $A:= (A_1,
A_2)$. 

\textit{Remark 2.2:} For the reconstruction mappings, we excluded the action
sequence as an input since $A^n$ is a function of the other input $M$. In our
(information) rate expressions, we will see the appearance of $A$ in the
expressions. As we will see in the next subsection, an advantage of this
definition is that it carries over to the case when the encoder takes actions
rather than the decoders.
\subsection{Source coding with action taken at the encoder}
This setting is shown in figure~\ref{fig:3}. As the definitions and problem 
statement for this case are similar to the first setting, we will only mention
the differences between the two settings. The main difference is that the
encoder takes actions rather than the decoders. Therefore, in the definition of
a code, we replace the case of a joint action encoder at the
decoders with the \textit{encoder taking actions} given by the function
\begin{align*}
f_{A-Enc} &: \Xc^n \rightarrow \mathcal{A}^n.
\end{align*}
As in the setting of actions taken at the decoder, here too we assume that the side information observed by the decoders 
is not available at the encoder. In subsequent sections we also describe the
results pertaining to the case where side information is available at the
encoder. 

\textit{Remark 2.3: Lossless source coding - } Some of our results concern the
case of 
lossless source coding. In the case of lossless source coding, the definitions
are similar, except that the distortion constraints $D_1, D_2$ are replaced by
the block probability of error constraint: $\P(\{\Xh_1^n \neq X^n\} \cup
\{\Xh_2^n \neq X^n\}) \le \e$. 

\section{Lossless source coding with actions at the decoders} \label{sect:3}
In this section and the next, we consider the case of source coding with
actions 
taken at the decoders. We first present results for the lossless source coding
setting. While the lossless case can be taken to be a special case of lossy
source coding, we present them separately, as we are able to obtain stronger
results for more general scenarios in the lossless setting, and give several
interesting examples that arise from this setup. The case of lossy source coding
for two decoders is presented in section \ref{sect:3.2}.

For the lossless case, we first state the result for the general case of $K$
decoders. Our result is stated in Theorem~\ref{thm:1}.
\begin{theorem} \label{thm:1}
Let the action channel be given by the conditional distribution $P_{Y_1, Y_2, \ldots, Y_K|X,A}$ with decoder $j$ observing the side information $Y_{j}$. Then, the minimum rate required for lossless source coding with actions taken at the decoders and cost constraint $C$ is given by
\begin{align*}
R &= \min \max_{j \in [1:K]} \left \{H(X|Y_j, A) \right\} + I(X;A),
\end{align*}
where $\min$ is taken over the distributions $p(x)p(a|x)p(y_1, y_2,
\ldots, y_K|x,a)$ such that $\E \Lambda(A) \le C$.
\end{theorem}

\subsection*{Achievability}
As the achievability techniques used are fairly standard (cf.
\cite{El-Gamal--Kim2010}), we give only a sketch of achievability.

\noindent \textit{Codebook Generation:}
\begin{itemize}
\item Generate $2^{n(I(X;A) + \e)}$ $A^n$ sequences according to $\prod_{i=1}^n p(a_i)$.
\item Bin the set of all $X^n$ sequences into $2^{n(\max_{j \in [1:K]} \left \{H(X|Y_j, A) \right\} + \e)}$ bins, $\Bc(m_b)$, $m_b \in [1:2^{n(\max_{j \in [1:K]} \left \{H(X|Y_j, A) \right\} + \e)}]$.
\end{itemize}

\noindent \textit{Encoding:}
\begin{itemize}
\item Given a source sequence $x^n$, the encoder looks for an index $M_A \in [1: 2^{n(I(X;A) + \e)}]$ such that $(x^n, a^n(M_A)) \in \aep$. If there is none, it outputs an uniform random index from $[1: 2^{n(I(X;A) + \e)}]$. If there is more than one such index, it selects an index uniformly at random from the set of feasible indices. From the covering lemma \cite[Chapter 3]{El-Gamal--Kim2010}, the probability of error for this step goes to $0$ as $n \to \infty$ since there are $2^{n(I(X;A) + \e)}$ $A^n$ sequences. 
\item The encoder also looks the index $m_b \in [1:2^{n(\max_{j \in [1:K]} \left \{H(X|Y_j, A) \right\} + \e)}]$ such that $x^n \in \Bc(m_b)$.
\item It then sends the indices $m_b$ and $M_A$ to the decoders via the common link. This step requires a rate of $R = \max_{j \in [1:K]} \left \{H(X|Y_j, A) \right\} + I(X;A) + 2\e$.
\end{itemize}

\noindent \textit{Decoding:}
\begin{itemize}
\item The decoders take the joint action $a^n(M_A)$ and obtain their side informations $Y_j$ for $j \in [1:K]$.
\item Decoder $j$ then looks for the \textit{unique} $X^n$ sequence in bin $\Bc(m_b)$ such that $(X^n, Y^n_j, a^n(M_A)) \in \aep$. An error is declared if there is none more than one $x^n$ sequence satisfying the decoding condition. The probability of error for this step goes to $0$ as $n \to \infty$ from the strong law of large numbers and the fact that $|\Bc| > 2^{n(\max_{j \in [1:K]} \left \{H(X|Y_j, A) )\right\}}$.
\end{itemize}

\subsection*{Converse}

Given a $(n, 2^{nR}, C)$ code, consider the rate constraint for decoder $j$. We have {\allowdisplaybreaks
\begin{align*}
& nR \ge H(M) \\
& = I(M;X^n) \\
& \stackrel{(a)}{=} I(A^n;X^n) + I(M;X^n|A^n) \\
& \stackrel{(b)}{\ge} I(A^n;X^n) + H(M|A^n, Y_j^n) - H(M|A^n, X^n, Y_j^n) \\
& = H(X^n) - H(X^n|A^n) + I(M;X^n |A^n, Y_j^n) \\
& \stackrel{(c)}{\ge} H(X^n) - H(X^n|A^n) + H(X^n |A^n, Y_j^n) -n\e_n\\
& = H(X^n) - H(X^n|A^n) + H(X^n  |A^n) \\ &\qquad + H(Y_j^n|X^n, A^n) - H(Y_j^n|A^n)-n\e_n \\
& \stackrel{(d)}{\ge} \sum_{i=1}^n H(X_i) + H(Y_{ji}|X_i, A_i) - H(Y_{ji}|A_i)-n\e_n. 
\end{align*}}
$(a)$ follows from $A^n$ being a function of $M$; $(b)$ follows from the Markov chain $M \to (X^n, A^n) \to Y_{j}^n$; $(c)$ follows from the assumption of lossless source coding; $(d)$ follows from conditioning reduces entropy and the fact that the action channel is a discrete memoryless channel (DMC). Define $Q$ as the standard time sharing random variable. Observe that $H(X_Q|Q) = H(X_Q) = H(X)$, $H(Y_{jQ}|A_Q, X_Q, Q) = H(Y_{jQ}|A_Q, X_Q) = H(Y_j|A,X)$ and $H(Y_{jQ}|A_Q, Q) \le H(Y_j|A)$. Hence, we can write the lower bound as
\begin{align*}
nR &\ge n (H(X) + H(Y_j|X,A) - H(Y_j|A) - \e_n) \\
& = n(I(X;A) + H(X|Y_j,A) - \e_n). 
\end{align*}
Taking the intersection of all lower bounds for all $K$ decoders then give us the rate expression given in the Theorem. Finally, the cost constraint on the action follows from $C \ge \E\frac{1}{n}\sum_{i=1}^n \Lambda(A_i) = \E \Lambda(A)$.

We now specialize the result in Theorem~\ref{thm:1} to the case of source coding with switching dependent side information mentioned in the introduction. We consider the more general setting involving $K$ decoders.

\begin{corollary} \label{coro:1}
\textit{Source coding with switching dependent side information and no cost
constraints.} 
Let $(X,Y)$ be jointly distributed according to $p(x,y)$. Let  $\mathcal{A} =
[1:K]$ and $P_{Y_1, Y_2, \ldots, Y_K|X,A}$ be defined by $Y_j = Y$ when $A =j$
and $e$ otherwise for $j \in [1:K]$. Let $\Lambda(A):=0$ for all $a \in
\mathcal{A}$. Then, the minimum rate is given by 
\begin{align*}
 H(X|Y) + \frac{K-1}{K}I(X;Y).
\end{align*} 
\end{corollary}
\begin{proof}

Proof of Corollary \ref{coro:1} amounts to an explicit characterization of the distribution of $p(a|x)$ in Theorem~\ref{thm:1}. For each $j \in [1:K]$, we have, from Theorem \ref{thm:1},
\begin{align}
R &\ge H(X|Y_j, A) + I(X;A) \nonumber \\
& = H(X|Y) + I(X;Y) - I(X;Y_j|A). \label{eqn:1}
\end{align} 
Consider now the sum
\begin{align}
\sum_{j=1}^K I(X;Y_j|A) & \stackrel{(a)}{=} \sum_{a \in \mathcal{A}} p(a) I(X;Y|A = a) \nonumber \\
& = H(Y|A) - H(Y|X,A) \nonumber \\
& \stackrel{(b)}{\le} H(Y) - H(Y|X) \nonumber \\
& = I(X;Y). \label{eqn:2}
\end{align}
$(a)$ follows from the fact that $Y_j = e$ for $a \neq j$ and $Y_j = Y$ for $a = j$. $(b)$ follows from the Markov Chain $A-X-Y$. 

Next, summing over the $K$ lower bounds in \eqref{eqn:1}, we obtain
\begin{align*}
R &\ge \frac{1}{K}(KH(X|Y) + KI(X;Y) - \sum_{j=1}^K I(X;Y_j|A)) \\
& \ge H(X|Y) + I(X;Y) - \frac{1}{K} I(X;Y) \\
& = H(X|Y) + \frac{K-1}{K} I(X;Y),
\end{align*}
where we used inequality \eqref{eqn:2} in the second last step. 
Finally, noting that this lower bound on the achievable rate can be obtained from Theorem \ref{thm:1} by setting $A \bot X$ and $p(a = j) = 1/K$ completes the proof of Corollary \ref{coro:1}. 
\end{proof}

\textit{Remark 3.1:} The action can be set to a fixed sequence independent of the source sequence. This is perhaps not surprising since there is no cost on the actions.

\textit{Remark 3.2:} For $K = 2$ and $X = Y$, which is the example given in the introduction, we have $R = H(X)/2$.

\textit{Remark 3.3:} For this class of channels, the achievability scheme in Theorem \ref{thm:1} has a simple and interesting ``modulo-sum'' interpretation. We present a sketch of an alternative scheme for this class of switching channels for $K = 2$. It is straightforward to extend the achievability scheme given below to $K$ decoders.

\textit{Alternative achievability scheme }

Split the $X^n$ sequence into 2 equal parts; $X_{1}^{n/2}$ and $X_{n/2+1}^n$ and select the fixed action sequence of letting decoder 1 observe $Y_1^{n/2}$ and decoder 2 observe $Y_{n/2+1}^n$. Separately compress each part using standard random binning with side information to obtain $M_1 \in [1:2^{n(H(X|Y)/2 + \e)}]$ and $M_2\in [1:2^{n(H(X|Y)/2 + \e)}]$ corresponding to the first and second half respectively. Within each bin, with high probability, there are only $2^{nI(X;Y)/2}$ typical $X^{n/2}$ sequences and we represent each of them with an index $M_{j1} \in [1:2^{n(I(X;Y)/2 + \e)}]$, where $j \in \{1,2\}$. Send out the indexes $M_1$ and $M_2$, which requires a rate of $H(X|Y) + 2\e$. Next, send out the index $M_{11} \oplus M_{21}$ which requires a rate of $I(X;Y)/2 + \e$. From $M_1$ and side information $Y_{1}^{n/2}$, decoder 1 can recover $X_{1}^{n/2}$ with high probability. Therefore, it can recover $M_{11}$ with high probability. Hence, it can recover $M_{21}$ from $M_{11} \oplus M_{21}$ and therefore, recover the $X_{\frac{n}{2}+1}^n$ sequence. The same analysis holds for decoder 2 with the indices interchanged. 

Corollary~\ref{coro:2} gives the characterization of the achievable rate for a general switching dependent side information setup with cost constraint on the actions for two decoders.

\begin{corollary} \label{coro:2}
\textit{General switching depedent side information for 2 decoders.} Define the action channel as follows: $A \in \{0,1,2,3\}$; $A = 0, Y_1 = e, Y_2 = e$; $A = 1, Y_1 = Y, Y_2 = e$; $A = 2, Y_1 = e, Y_2 = Y$; and $A = 3, Y_1 = Y, Y_2 = Y$. Let $\Lambda(A = j) = C_j$ for $j \in [0:3]$. Then, the optimum rate-cost trade-off for this class of channel is given by
\begin{align*}
&R \ge I(X;A) + \max\left\{ H(X|Y_1,A), H(X|Y_2,A)\right\} \\
&= I(X;A) + p_0H(X|A = 0) + \sum_{j=1}^3 p_j H(X|Y, A = j) \\
& \qquad + \max\{p_1 I(X;Y| A = 1), p_2 I(X;Y|A = 2)\},
\end{align*}
for some $p(a|x)$, where $\P\{A = j\} = p_j$, satisfying $\sum_{j=0}^3 p_jC_j \le C$.
\end{corollary}

\textit{Remark 3.4:} This setup again has a ``modulo-sum interpretation'' for the term $\max\{p_1 I(X;Y|A = 1), p_2 I(X;Y| A = 2)\}$ and the rate can also be achieved by extending the achievability scheme described in Corollary~\ref{coro:1}. The scheme involves partitioning the $X^n$ sequence according to the value of $A_i$ for $i \in [1:n]$. Following the scheme in Corollary \ref{coro:1}, we let $M_j\in [1:2^{n (p_j H(X|Y, A=j) + \e)}]$ for $j \in [0:3]$. We first generate a set of $A^n$ codewords according to $\prod_{i=1}^n p(a_i)$. Next, for each $A^n$ codeword, define $A^{n_j}$ to be $\{A_i: A_i =j\}$. Similarly, let $X^{n_j} := \{X_i: A_i =j, i \in [1:n]\}$ be the set of possible $X$ sequences corresponding to $A^{n_j}$. We bin the set of all $X^{n_j}$ sequences to $2^{n(p_j H(X|Y, A=j) + \e)}$ bins, $\Bc_j(M_j)$. For $j \in \{1,2\}$, further bin the set of $x^{n_j}$ sequences into $2^{n(p_j I(X;Y| A = j) + \e)}$ bins, $\Bc_{j1}(M_{j1})$, $M_{j1} \in [1:2^{n(p_j I(X;Y| A = j) + \e)}]$. 

For encoding, given an $x^n$ sequence, the encoder first finds an $A^n$ sequence that is jointly typical with $x^n$. It sends out the index corresponding to the $A^n$ sequence found. Next, it splits the $x^n$ sequence into four partial sequences, $x^{n_j}$, for $j \in [0:3]$, where $x^{n_j}$ is the set of $x_i$ corresponding to $A_i = j$. It then finds the corresponding bin indices such that $x^{n_j} \in \Bc_j(I_j)$ for $j \in [0:3]$. It then sends out the indices $M_0, M_1, M_2, M_3$ and $M_{11} \oplus M_{21}$. 

For decoding, we mention only the scheme employed by the first decoder, since the scheme is the same in for decoder 2. From the properties of jointly typical sequences and standard analysis for Slepian-Wolf lossless source coding~\cite{Cover1975}, it is not difficult to see that decoder 1 can recover $x^{n_0}, x^{n_1}, x^{n_3}$ with high probability. Recovery of $x^{n_1}$ also allows decoder 1 to recover the index $M_{11}$ and hence, $M_{21}$ from $M_{11} \oplus M_{21}$. Noting that the rate of $M_{21}$ and $M_{2}$ sums up to $p_2 H(X|A=2) + 2\e$, it is then easy to see that decoder 1 can recover $x^{n_2}$ with high probability.    

In corollary \ref{coro:1}, we showed that, for the case of switching dependent side information, the action sequence is independent of the source $X^n$ when cost constraint on the actions is absent. A natural question to ask is whether the action is still independent of $X^n$ when a cost constraint on the actions is present? The following example shows that the optimum action sequence is in general dependent on $X^n$.

\textit{Example 1: Action is dependent on source statistics when cost constraint is present.} 
Let $K = 2$ and $(X, Y)$ be distributed according to an $S$ channel, with $X \sim \Bern(1/2)$, $\P(Y=1|X=1) = 1$ and $\P(Y=0|X=0) = 0.2$. Let $A \in \{1,2\}$ with $Y_1 = Y$ if $A = 1$ and $Y_2 = Y$ if $A=2$. Let $\P(A = 1) = p_1$, $\P(X = 0|A = 1) = 1/2 + \d_1$ and $\P(X = 0|A = 2) = 1/2 - \d_2$. Figure \ref{fig:ex} shows the probability distributions between the random variables. 

\begin{figure}[!h]
\begin{center}
\psfrag{x}[c]{$X$}
\psfrag{y}[c]{$Y$}
\psfrag{a}[c]{$A$}
\psfrag{0}[c]{$0$}
\psfrag{1}[c]{$1$}
\psfrag{2}[c]{$2$}
\psfrag{pxy}[c]{$0.2$}
\psfrag{pax1}[c]{$\frac{1}{2} + \d_1$}
\psfrag{pax2}[c]{$\frac{1}{2} + \d_2$}
{\includegraphics{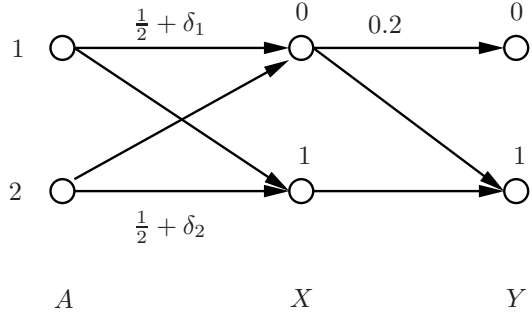}} 
\end{center}
\caption{Probability distributions for random variables used in example 1} \label{fig:ex}
\vspace{-5pt}
\end{figure}
 
Since $X\sim \Bern(1/2)$, $d_1$ and $d_2$ are related by $\d_2 = p_1\d_1/(1-p_1)$. Therefore, we set $\d_1 = \d$ and $\kappa = p_1/(1-p_1)$ for this example. 

Now, let $\Lambda(A = 1) = 1$ and $\Lambda(A = 2) = 0$ and $C = 0.4$. The optimum rate-cost tradeoff in this case may be obtained from Corollary \ref{coro:2} by setting $C_0 = C_3 = \infty$, $C_1 = 1$ and $C_2 = 0$, giving us
\begin{align*}
R &= I(X;A)+ p_1 H(X|Y, A = 1) + (1-p_1) H(X|Y,A=2) \\
&\qquad + \max\{p_1 I(X;Y| A = 1), (1-p_1)I(X;Y|A = 2)\},
\end{align*}
for some $p(a|x)$, where $\P\{A = 1\} = p_1$, satisfying $p_1 \le 0.4$. The problem of finding the optimum action sequence to take then reduces (after some straightforward algebra) to the following optimization problem:
\begin{align*}
&\min_{p_1, \d}  \; 1 - p_1 H_2(0.5 - \d) - (1- p_1)H_2(0.5 - \kappa\d)  \\
&   \qquad + p_1 H(X|Y, A = 1) + (1-p_1) H(X|Y,A=2) \\
&  \qquad + \max\{p_1 I(X;Y| A = 1), (1-p_1)I(X;Y|A = 2)\},  \\
&\mbox{subject to } \\
&  \qquad 0 \le p_1 \le 0.4, \\
&  \qquad -0.5 \le \d \le 0.5,
\end{align*}
where
\begin{align*}
H(X|A=1) &= p_1H_2(0.5-\d), \\
H(X|A=2) &= (1-p_1)H_2(0.5 - \kappa\d), \\
H(X|Y, A = 1) & = ((0.5+\d)(0.8) + 0.5-\d)H_2\left(\frac{0.5-\d}{(0.5+\d)(0.8) + (0.5-\d)}\right), \\
H(X|Y, A = 2) &= ((0.5-\kappa\d)(0.8) + (0.5+ \kappa\d))H_2\left(\frac{0.5+\kappa\d}{(0.5-\kappa\d)(0.8) + (0.5+\kappa\d)}\right),
\end{align*}
and $H_2(.)$ is the binary entropy function.

While exact solution to this (non-convex) optimization problem involves searching over $p_1$ and $\d$, it is easy to see that if $A$ is restricted to be independent of $X$, which corresponds to restricting $\d$ to be equal to $0$, then the optimum solution for $p_1$ is $0.4$. Under $p_1 =0.4$ and $\d = 0$, we obtain $R_{A\bot X} = 0.9568$. In contrast, setting $p_1 =0.4$ and $\d = -0.05$, we obtain $R = 0.9554$, which shows that the optimum action sequence is in general dependent on the source $X$ when cost constraints are present.

An explanation for this observation is as follows. The cost constraint forces decoder 1 to see less of the side information $Y$ than decoder 2. It may therefore make sense to bias the distribution $X|A=1$ so that $Y$ conveys more information about the source sequence $X$, even at the expense of describing the action sequence to the decoders. Roughly speaking, the amount of information conveyed about $X$ by $Y$ may be measured by $I(X;Y)$. Note that under $\d = 0$, $I(X;Y|A=1) = 0.108$, whereas under $\d = -0.05$, $I(X;Y|A=1) = 0.1116$. A plot of the optimum rate versus cost tradeoff obtained by searching over a grid of $p_1$ and $\d$ is shown in Figure \ref{fig:ratevcost}. The figure also shows the rate obtained if actions were forced to be independent of the source sequence.

\begin{figure}[!h]
\begin{center}
\scalebox{0.7}{\includegraphics{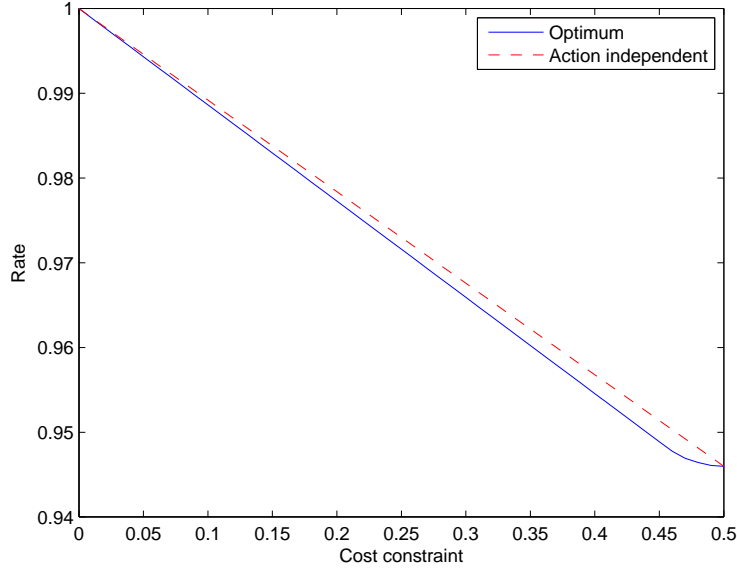}}
\caption{Rate versus cost constraint for the example 1. It is easy to show operationally that the optimum rate versus cost curve is convex in the cost constraint. When the cost constraint approaches zero, the rate approaches 1, since this case corresponds to decoder 1 not seeing any of the side information. When the cost constraint approaches 0.5, the rate approaches the minimum rate without cost constraint. The red dashed line shows the rate that would be obtained if actions were forced to be independent of the source. As can be seen on graph, forcing actions to be independent of the source is in general not optimum when cost constraint is present. The optimum rate versus cost constraint plot appears to be linear over a range of cost constraints. It can be shown that if the cost constraint is below a threshold, then the optimum rate is a linear function of the cost constraint. However, the plot obtained via numerical simulation appears to be linear in the cost constraint over a wider range than what we obtained by analysis. Performing a more refined analysis to obtain a cost constraint threshold that matches the cost threshold obtained by simulation appears to be difficult, due to the nature of the optimization problem that is involved.} \label{fig:ratevcost}
\end{center}
\end{figure}


\section{Lossy source coding with action at the decoders} \label{sect:3.2}
In this section, we first consider the case when causal reconstruction is required, and give the general rate-distortion-cost region for $K$ decoders. Next, we consider the case of lossy noncausal reconstruction for two decoders and give a general achievability scheme for this case. We then show that our achievability scheme is optimum for several special cases. Finally, we discuss some connections between our setting and the complementary delivery setting introduced in~\cite{Kimura2006}.

\subsection{Causal reconstruction for $K$ decoders}

\begin{theorem} \label{thm:3}
\textit{Causal lossy reconstruction for $K$ decoders}

When the decoders are restricted to causal reconstruction \cite{Weissman--El-Gamal06}, $\mathcal{R}(D_1, D_2,\ldots, D_K, C)$ is given by
\begin{align*}
R = I(U;X)
\end{align*} 
for some $p(u|x)$, $A = f(U)$ and reconstruction functions $\xh_j$ for $j \in [1:K]$ such that
\begin{align*}
\E d_j(X, \xh_j(U, Y_j)) &\le D_j \mbox{ for } j \in [1:K] \\
\E \Lambda(A) &\le C.
\end{align*}
The cardinality of $U$ is upper bounded by $|\Uc| \le |\Xc||\Ac| + K$.
\end{theorem}
\textit{Remark 4.1}: Theorem \ref{thm:3} generalizes the corresponding result for one decoder in \cite[Theorem 3]{Permuter2010a}.

\begin{proof}
As the achievability scheme is a straightforward extension of the scheme in \cite[Theorem 3]{Permuter2010a}, we will omit the proof of achievability here. For the converse, given a code that satisfies the cost and distortion constraints, we have
\begin{align*}
nR &\ge H(M) \\
&= I(X^n;M) \\
& \stackrel{(a)}{=} \sum_{i=1}^n(H(X_i) - H(X_i|M, X^{i-1})) \\
&\stackrel{(b)}{= }\sum_{i=1}^n(H(X_i) - H(X_i|M, X^{i-1}, A^{i-1})) \\
&\stackrel{(c)}{= }\sum_{i=1}^n(H(X_i) - H(X_i|M, X^{i-1}, A^{i-1}, Y_1^{i-1}, \ldots, Y_K^{i-1})) \\
 &\ge \sum_{i=1}^n(H(X_i) - H(X_i|U_i)), 
\end{align*}
where $(a)$ follows from the fact that $X^n$ is a memoryless source; $(b)$ follows from the fact that $A^{i-1}$ is a function of $M$; $(c)$ follows from the fact that the action channel $p(y_1, y_2, \ldots, y_k|x, a)$ is a memoryless channel; and the last step follows from defining $U_i = (M, Y_1^{i-1}, \ldots, Y_K^{i-1})$. Finally, defining $Q$ to be a random variable uniform over $[1:n]$, independent of all other random variables, $U = (U_Q,Q)$, $X = X_Q$, $A = A_Q$ and $Y_j = Y_{jQ}$ for $j \in [1:K]$ then gives the required lower bound on the minimum rate required. Further, we have $A = f(U)$. It remains to verify that the cost and distortion constraints are satisfied. Verification of the cost constraint is straightforward. For the distortion constraint, we have for $j \in [1:K]$
\begin{align*}
\E \frac{1}{n} \sum_{i=1}^n d_j(X_i, \xh_{ji}(M, Y_j^{i})) \ge \E d_j(X, \xh'_{j}(U,Y_j)),
\end{align*} 
where we define $\xh_{j}'(U,Y_j):= \xh_{jQ}(M, Y_j^i)$. This shows that the definition of the auxiliary random variable $U$ satisfies the distortion constraints. Finally, the cardinality of $U$ can be upper bounded by using the support lemma \cite{Eggleston}. We require $|\Xc||\Ac|-1$ letters to preserve $P_{X,A}$, which also preserves the cost constraint. In addition, we require $K+1$ letters to preserve the rate and $K$ distortion constraints.
\end{proof}
We now turn to the case of noncausal reconstruction. For this setting, we give results only for the case of two decoders.
\subsection{Noncausal reconstruction for two decoders}

We first give a general achievability scheme for this setting.
\begin{theorem} \label{thm:2}
An achievable scheme for the lossy source coding with actions at the decoders is given by
\begin{align*}
R &\ge I(X;A) + \max\left\{I(X;U|A, Y_1), I(X;U|A,Y_2)\right\} \\
& \qquad + I(X;V_1|U, A, Y_1) + I(X;V_2|U,A,Y_2)
\end{align*}
for some $p(x)p(a|x)p(u|a,x)p(v_1|u,a,x)p(v_2|u,a,x)p(y_1, y_2|x,a)$ and reconstruction functions $\xh_1$ and $\xh_2$ satisfying
\begin{align*}
\E d_j(X, \xh_j(U, V_j, A, Y_j)) &\le D_j \mbox{ for } j = 1,2, \\
\E \Lambda(A) &\le C.
\end{align*}
\end{theorem}
We provide a sketch of achievability in Appendix \ref{appen1} since the techniques used are fairly straightforward. As an overview, the encoder first tells the decoders the action sequence to take. It then sends a common description of $X^n$, $U^n$, to both decoders. Based on the action sequence $A^n$ and the common description $U^n$, the encoder sends $V_1^n$ and $V_2^n$ to decoders 1 and 2 respectively. We do not require decoder 1 to decode $V_2^n$, or for decoder 2 to decode $V_1^n$. 

Theorem~\ref{thm:2} is optimum for the following special cases.

\begin{proposition} \label{prop:HB} \textit{Heegard-Berger-Kaspi~\cite{Heegard--Berger1985}, \cite{Kaspi1994} Extension.} Suppose the following Markov chain holds: $(X,A) - (A,Y_1) - (A, Y_2)$. Then, the rate-distortion-cost trade-off region is given by 
\begin{align*}
R &\ge I(X;A) + I(X;U|A, Y_2) \\
& \qquad + I(X;V_1|U, A, Y_1)
\end{align*}
for some $p(x)p(a|x)p(u,v_1|x,a)p(y_1|x,a)p(y_2|y_1,a)$ satisfying
\begin{align*}
\E d_1(X, \Xh_1(U, V_1, A, Y_1)) &\le D_1, \\
\E d_2(X, \Xh_2(U, A, Y_2)) &\le D_2, \\
\E \Lambda(A) &\le C.
\end{align*}
The cardinality of the auxiliary random variables is upper bounded by $|\mathcal{U}| \le |\Xc||\mathcal{A}|+2$ and $|V_1| \le |\mathcal{U}|(|\Xc||\mathcal{A}|+1)$. 
\end{proposition}
The achievability for this proposition follows from Theorem~\ref{thm:2} by setting $V_2 = \emptyset$ and noting that since $(X,A) - (A,Y_1) - (A, Y_2)$, the terms in the $\max\{.\}$ function simplifies to $I(X;U|A, Y_2)$. We give a proof of converse as follows.

\begin{proof}[Converse]
Given a code that satisfies the constraints, {\allowdisplaybreaks
\begin{align*}
nR & \ge H(M) \\
& = H(M, A^n) \\
& = H(A^n) + H(M|A^n) \\
& \ge H(A^n) - H(A^n|X^n) + H(M|A^n, Y_2^n) - H(M|Y_2^n, A^n, X^n) \\
& = I(X^n; A^n) + I(X^n;M|A^n, Y^n_2) \\
& = I(X^n;A^n) + I(X^n;M, Y_1^n|A^n, Y_2^n) - I(X^n;Y_1^n|M, A^n, Y_2^n)\\
& = I(X^n;A^n) + H(X^n|A^n, Y_2^n) - H(X^n|M, Y_1^n, A^n, Y_2^n) - I(X^n;Y_1^n|M, A^n, Y_2^n) \\
& = I(X^n;A^n) + H(X^n|A^n, Y_2^n) - \sum_{i=1}^n (H(X_i|M, Y_1^n, A^n, Y_2^n, X^{i-1}) + I(X^n;Y_{1i}|M, A^n, Y_2^n, Y_{1}^{i-1})) \\
& \ge I(X^n;A^n) + H(X^n|A^n, Y_2^n) - \sum_{i=1}^n (H(X_i|M, Y_1^n, A^n, Y_2^n) + I(X^n;Y_{1i}|M, A^n, Y_2^n, Y_{1}^{i-1})) \\
& \stackrel{(a)}{=} I(X^n;A^n) + H(X^n|A^n, Y_2^n) - \sum_{i=1}^n (H(X_i|M, Y_1^n, A^n, Y_2^n) + I(X_i;Y_{1i}|M, A^n, Y_2^n, Y_{1}^{i-1})) \\
& = I(X^n;A^n) + H(X^n|A^n, Y_2^n) - \sum_{i=1}^n H(X_i|M, Y_1^{i-1}, A^n, Y_2^n) \\
& \quad + \sum_{i=1}^n (I(X_i; Y_{1i}^n| M,  A^n, Y_2^n, Y_1^{i-1}) - I(X_i;Y_{1i}|M, A^n, Y_2^n, Y_{1}^{i-1}))\\
& = I(X^n;A^n) + H(X^n|A^n, Y_2^n) - \sum_{i=1}^n H(X_i|M, Y_1^{i-1}, A^n, Y_2^n) \\
& \quad + \sum_{i=1}^n (I(X_i; Y_{1,i+1}^{n}| M,  A^n, Y_2^n, Y_1^{i}), \\
& = I(X^n;A^n) + H(X^n|A^n, Y_2^n) - \sum_{i=1}^n H(X_i|M, Y_1^{i-1}, A^n, Y_2^n) \\
& \quad + \sum_{i=1}^n (I(X_i; Y_{1,i+1}^{n}| M,  A^n, Y_2^{n \backslash i}, Y_{1i},Y_1^{i-1}), 
\end{align*}
}where $(a)$ follows from the fact that $X^{n \backslash i} - (M, A^n, Y_2^n, Y_{1}^{i-1}, X_i) - Y_{1i}$ and the last step follows from the Markov Chain assumption $X_i - (A_i, Y_{1i}) - (A_i, Y_{2i})$.
Consider now,
\begin{align*}
I(X^n;A^n) + H(X^n|A^n, Y_2^n) & = I(X^n;A^n) + H(X^n, Y_2^n|A^n) - H(Y_2^n|A^n) \\
&= H(X^n) + H(Y_2^n|A^n, X^n) -H(Y^n_2|A^n) \\
& \ge \sum_{i=1}^n (H(X_i) + H(Y_{2i}|X_i, A_i) -H(Y_{2i}|A_i)). 
\end{align*}
Hence, 
\begin{align*}
nR &\ge \sum_{i=1}^n (H(X_i) + H(Y_{2i}|X_i, A_i) -H(Y_{2i}|A_i)) -  \sum_{i=1}^n H(X_i|M, Y_1^{i-1}, A^n, Y_2^n) \\
& \quad + \sum_{i=1}^n (I(X_i; Y_{1,i+1}^{n}| M,  A^n, Y_2^{n\i},Y_{1i} ,Y_1^{i-1}). 
\end{align*}
Define now $Q$ to be a random variable uniform over $[1:n]$, independent of all other random variables; $X = X_Q$, $Y_1 = Y_{1Q} $, $Y_{2} = Y_{2Q}$, $A = A_Q$, $U_i = (M, Y_1^{i-1}, A^{n\backslash i}, Y_2^{n \backslash i })$, $V_i = Y_{1,i+1}^n$, $U= (U_Q, Q)$ and $V = V_Q$. Then, we have
\begin{align*}
R &\ge H(X) + H(Y_{2}|X, A) -H(Y_{2}|A, Q) -  H(X|A, Y_2, U) \\
& \quad + I(X; V| A, Y_1, U)  \\
& \ge H(X) + H(Y_{2}|X, A) -H(Y_{2}|A) -  H(X|A, Y_2, U) \\
& \quad + I(X; V| A, Y_1, U)  \\
& = I(X;A) + I(X;U|A, Y_2) + I(X; V| A, Y_1, U).
\end{align*}
It remains to verify that the definitions of $U$, $V$ and $A$ satisfy the distortion and cost constraints, which is straightforward. Prove of the cardinality bounds follows from standard techniques.
\end{proof}
The next proposition extends our results for the case of switching dependent side information to the a class of lossy source coding with switching dependent side information. 

\begin{proposition}\textit{Special case of switching dependent side information.} \label{prop2}
Let $Y_1 = X, Y_2 = Y$ if $A = 1$ and $Y_1 = Y, Y_2 = X$ if $A = 2$ and for all $x$, there exists $\xh_1$ and $\xh_2$ such that $d_1(x,\xh_1) = 0$ and $d_2(x, \xh_2) = 0$. Then, the rate-distortion-cost trade-off region is given by 
\begin{align*}
R &\ge I(X;A) + \max\{\P(A=2) I(X;U_1|A=2, Y), \\
&\qquad \qquad \qquad \qquad \P(A=1)I(X;U_2|A = 1,Y)\}
\end{align*}
for some $p(x,y)p(a|x)p(u_1|x,a=2)p(u_2|x,a=1)$ satisfying
\begin{align*}
\P(A = 2)\E d_1(X, \Xh_1(Y, U_1)|A = 2) &\le D_1, \\
\P(A = 1)\E d_2(X, \Xh_2(Y, U_2)|A = 1) &\le D_2, \\
\E \Lambda(A) &\le C.
\end{align*}
The cardinality of the auxiliary random variables is upper bounded by $|\mathcal{U}_1|\le |\Xc|+1$ and $|\mathcal{U}_2|\le |\Xc|+1$.
\end{proposition}
Achievability follows from Theorem~\ref{thm:2} by setting $V_1 = V_2 = \emptyset$ and $U = U_2$ if $A = 1$ and $U = U_1$ if $A = 2$. We give the proof of converse as follows. 

\begin{proof}[Converse]
Given a code that satisfies the cost and distortion constraints, consider the rate required for decoder 1. We have
\begin{align*}
nR & \ge H(M) \\
& = H(M, A^n) \\
& = H(A^n) + H(M|A^n) \\
& \ge H(A^n) - H(A^n|X^n) + H(M|A^n, Y_1^n) - H(M|Y_1^n, A^n, X^n) \\
& = I(X^n; A^n) + I(X^n;M|A^n, Y^n_1) \\
& = I(X^n;A^n) + H(X^n, Y_1^n|A^n) - H(Y_1^n|A^n) - H(X^n|M, A^n, Y_1^n) \\
&= H(X^n) + H(Y_1^n|A^n, X^n) -H(Y^n_1|A^n) - H(X^n|M, A^n, Y_1^n)\\
& \ge \sum_{i=1}^n (H(X_i) + H(Y_{1i}|X_i, A_i) -H(Y_{1i}|A_i)) -  \sum_{i=1}^n H(X_i|M, A^n, Y_1^n). 
\end{align*}
As before, we define $Q$ to be an uniform random variable over $[1:n]$, independent of all other random variables. We then have
\begin{align*}
R &\ge H(X_Q|Q) + H(Y_{1Q}|X_Q,A_Q,Q) - H(Y_Q|A_Q,Q) - H(X_Q|M, A^n, Y_1^n,Q) \\
&\stackrel{(a)}{\ge} H(X) + H(Y_{1}|X,A) - H(Y_1|A) - H(X|M, A^n, Y_1^n) \\
& \stackrel{(b)}{=}  I(X;A) + I(X;U_1|Y_1,A).
\end{align*}
$(a)$ follows from the discrete memoryless nature of the action channel and the fact that conditioning reduces entropy; $(b)$ follows from defining $U_{1i} = (M, A^{n \backslash i}, Y_1^{n\backslash i})$ and $U_1 = (U_{1Q}, Q)$. Expanding the second term in terms of $A$ and using the observation that $Y_1 =X$ when $A = 1$ and $Y_1 = Y$ when $A = 2$, we obtain
\begin{align*}
R &\ge I(X;A) + \P(A = 2)I(X;U_1|Y, A=2).
\end{align*}
For decoder 2, the same steps with side information $Y_2$ instead of $Y_1$ and defining $U_{2i} = (M, A^{n \backslash i}, Y_2^{n\backslash i})$, $U_2 = (U_{2Q},Q)$ yield
\begin{align*}
R &\ge I(X;A) + \P(A = 1)I(X;U_2|Y, A=1).
\end{align*}
Taking the maximum over two lower bounds yield
\begin{align*}
R \ge I(X;A) + \max\{\P(A = 2)I(X;U_1|Y, A=2), \P(A = 1)I(X;U_2|Y, A=1)\}
\end{align*}
for some $p(a|x)p(u_1, u_2|x,a)$. Verifying the cost constraint is straightforward. As for the distortion constraint, we have for the decoder 1
\begin{align*}
\frac{1}{n}\E d_1(X^n, \xh_1^n(M, A^n, Y_1^n)) &= \E d_1(X, \xh_1(U_1, A, Y_1)) \\
& = \P(A=2) \E (d_1(X, \xh_1(U_1, Y))|A=2).
\end{align*} 
The same arguments hold for decoder 2. It remains to show that the probability distribution can be restricted to the form $p(a|x)p(u_1|a,x)p(u_2|a,x)$. Observe that $\P(A=2) \E (d_1(X, \xh_1(U_1, Y))|A=2)$ and $\P(A = 2)I(X;U_1|Y, A=2)$ depends on the joint distribution only through the marginal $p(a,u_1|x)$ and $\P(A=1) \E (d_2(X, \xh_2(U_2, Y))|A=1)$ and $\P(A = 1)I(X;U_2|Y, A=1)$ depends on the joint distribution only through the marginal $p(a,u_2|x)$. Hence, restricting the joint distribution to the form $p(a|x)p(u_1|a,x)p(u_2|a,x)$ does not affect the rate, cost or distortion constraints. It remains to bound the cardinality of the auxiliary random variables used, which follows from standard techniques. This completes the proof of converse.
\end{proof}
\textit{Remark 4.2:} The condition on the distortion constraints is simply to remove distortion offsets. It can be removed in a fairly straightforward manner. 

\textit{Remark 4.3:} As with the lossless source coding with switching dependent side information case, a modulo sum interpretation for the terms in the $\max$ expression is possible. When $A = 1$, the encoder codes for decoder 2, resulting, after binning, in an index $I_2$  for the codeword $U^n_2$; and when $A = 2$, the encoder codes for decoder 1, resulting, after binning, in an index $I_1$ for the codeword $U_1^n$. The encoder sends out the modulo sum of the indices of the two codewords ($I_1 \oplus I_2$) along with the index of the action codeword. Decoder 1 has the $X_i$ sequence when $A = 2$ and hence, it has the index $I_2$. Therefore, it can recover it's desired index $I_1$ from $I_1 \oplus I_2$. A similar analysis holds for decoder 2.

\textit{Example 2: Binary source with Hamming distortion and no cost constraint.} Let $Y = \emptyset$ and $X\sim \Bern(1/2)$. Assume no cost on the actions taken: $\Lambda(A = 1) = \Lambda(A = 2) = 0$ and let the distortion measure be Hamming. Then, the rate distortion trade-off evaluates to
\begin{align*}
R &= \min_{\alpha}\max\left\{\alpha\left(1 - H_2\left(D_1/\alpha \right)\right)\mathbf{1}\left(D_1/\alpha \le 1/2\right), \right. \\
&\quad \left. (1-\alpha)\left(1 - H_2\left(D_2/(1-\alpha)\right)\right)\mathbf{1}\left(D_2/(1-\alpha) \le 1/2\right)\right\},
\end{align*}
where $\mathbf{1}(x)$ denotes the indicator function. As a check, note that if $D_1, D_2 \to 0$, then the rate obtained is $1/2$, which agrees with the rate obtained in Corollary \ref{coro:1} for the lossless case. The result follows from explicitly evaluating the result in Proposition \ref{prop2}. Let $\P(A=2) = \alpha$. From Proposition 2, we have
\begin{align*}
R &\ge I(X;A) + \P(A = 2)I(X;U_1|Y, A=2) \\
& = 1 - (1-\alpha)H(X|A=1) - \alpha H(X|A=2) +\alpha H(X|A=2) - \alpha H(X|U_1, A=2) \\
& \ge \alpha -\alpha H(X|U_1, A=2) \\
& \ge \alpha(1 - H(X \oplus \Xh_1|U_1, A=2)) \\
& \ge \alpha\left(1 - H_2\left(\frac{D_1}{\alpha}\right )\right)\mathbf{1}\left(\frac{D_1}{\alpha} \le 1/2\right). 
\end{align*}
The last step follows from the observations that (i) if $D_1/\alpha > 1/2$, then we lower bound $R$ by $0$; and (ii) if $D_1/\alpha \le 1/2$, then from the distortion constraint $\alpha \E d(X, \Xh_1|A=2) \le D_1$, $H(X \oplus \Xh_1|A=2) \le H_2(D_1/\alpha)$. The other bound is derived in the same manner. The fact that this rate can be attained is straightforward, since we can choose $U_1 = \Xh_1$ when $A=2$ and $U_2 = \Xh_2$ when $A=1$.
In this example, the action sequence is independent of the source, but unlike the case of lossless source coding, $\P(A=1)$ is not in general equal to $\P(A=2)$. It depends on the distortion constraints for the individual decoders. A surface plot of the rate versus distortion constraints for the two decoders is shown in Figure \ref{fig:bin_dist}. 

\begin{figure}[!h]
\begin{center}
\scalebox{0.4}{\includegraphics{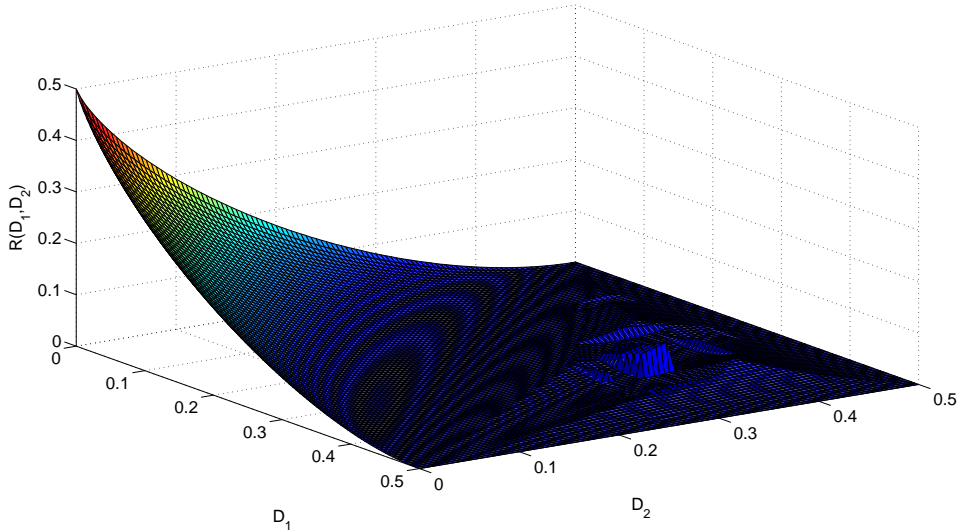}}
\caption{Plot of rate versus distortions. The figure above plots the rate distortion surface $R(D_1,D_2)$ for the Example 2.  There is no side information, i.e., $Y = \emptyset$ and $X \sim \Bern(1/2)$. Assume no cost on the actions taken: $\Lambda(A = 1) = \Lambda(A = 2) = 0$ and let the distortion measure be Hamming. Note that if any of $D_1,D_2 \rightarrow 0.5$, $R$ approaches $0$, also if $D_1=D_2=0$, rate is $0.5$} \label{fig:bin_dist}
\end{center}
\end{figure}  

\subsection{Connections with Complementary Delivery}
In the prequel, we consider several cases for switching dependent side
information in 
which the achievability scheme has a simple ``modulo sum'' interpretation for
the terms in the $\max$ function. This interpretation is not unique to our setup
and in this subsection, we consider the complementary delivery setting
\cite{Kimura2006} in which this interpretation also arises. Formally, the
complementary delivery problem is a special case of our setting and is obtained
by letting $A = \emptyset$, $X = (\Xt, \Yt)$, $\P(Y_1, Y_2|X) = 1_{Y_1 = \Xt,
Y_2 = \Yt}$, $\Lambda(A) = 0$, $d_1(X, \Xh_1) = d'_1(\Yt, \Xh_1)$ and $d_2(X,
\Xh_2) = d'_2(\Xt, \Xh_2)$. For this subsection, for notational convenience, we
will use $X$ in place of $\Xt$, $Y$ in place of $\Yt$, $\Yh$ in place of $\Xh_1$
and $\Xh$ in place of $\Xh_2$. This setting is shown in Figure~\ref{fig:comdel}. 
\begin{figure}[!h]
\psfrag{xy}[c]{$(X^n, Y^n)$}
\psfrag{x}[c]{$X^n$}
\psfrag{y}[c]{$Y^n$}
\psfrag{r}[c]{$R$}
\psfrag{yh}[c]{$\Yh^n$}
\psfrag{xh}[c]{$\Xh^n$}
\psfrag{d1}[c]{Decoder 1}
\psfrag{d2}[c]{Decoder 2}
\psfrag{e}[c]{Encoder}
\begin{center}
\scalebox{0.85}{\includegraphics{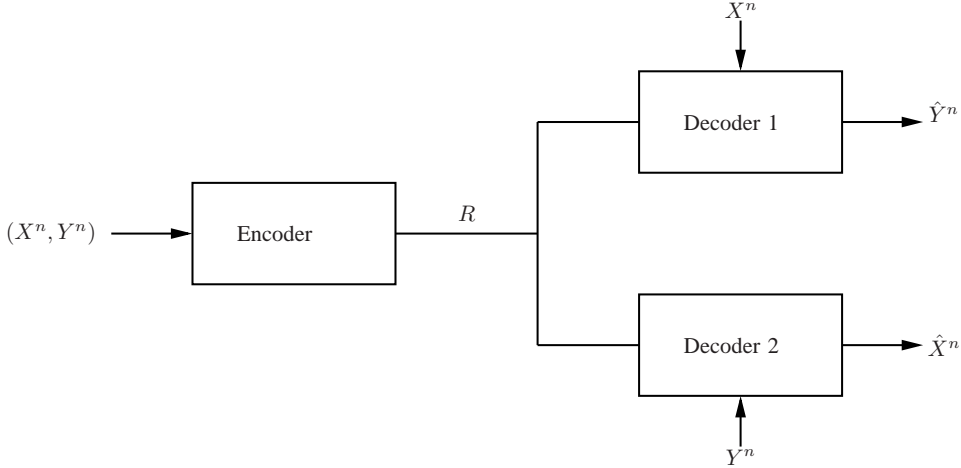}}
\end{center}
\caption{Complementary Delivery setting} \label{fig:comdel}
\end{figure}
In~\cite{Kimura2006}, the following achievable
rate was established{\allowdisplaybreaks
\begin{align}
R(D_1, D_2) \ge \max\{ I(U;Y|X), I(U;X|Y) \}, \label{eqn:com_del}
\end{align}
for some $p(u|x,y)$ satisfying $\E d_1(Y, \Yh(U,X)) \le D_1$ and $\E d_2(X, \Xh(U,Y)) \le D_2$.} 

Our achievability scheme in Theorem~\ref{thm:2} generalizes this scheme when
specialized 
to the complementary delivery setting, but we do not yet know if our achievable
rate can be strictly smaller for the same distortions. However, by taking a
modulo sum interpretation for the terms in the $\max\{.\}$ function in
(\ref{eqn:com_del}), as we have done for several examples in this paper, we are
able to give simple proofs and explicit characterization for two canonical cases: the
Quadratic Gaussian and the doubly symmetric binary Hamming distortion
complementary delivery problems. While characterizations for these two settings also appear independently in \cite{Timo}, our approach in characterizing these settings is different from that in \cite{Timo}, and we believe would be of interest to readers. Furthermore, by taking the ``modulo sum''
interpretation, we establish the following, which may be
a useful observation in practice: ``\textit{For the Quadratic Gaussian
complementary delivery problem, if one has a good code (in the sense of
achieving the optimum rate distortion tradeoff) for the point to point Wyner-Ziv
\cite{Wyner} Quadratic Gaussian setup, then a simple modification exists to turn
the code into a good code for the Quadratic Gaussian complementary delivery
problem.}'' A similar observation holds for the doubly symmetric binary Hamming
distortion case. We first consider the Quadratic Gaussian case.

\begin{proposition} \textit{Quadratic Gaussian complementary delivery.} 
Let $Y = X+Z$, where $Z \sim N(0, N)$ is independent of $X\sim N(0,P)$, and the
distortion measures be mean square distortion. Let $P' = PN/(P+N)$. The rate
distortion region for the non-trivial constraints of $D_2 \le P'$ and $D_1 \le
N$ is given by
\begin{align*}
R(D_1, D_2) = \max\left\{\frac{1}{2}\log\left(\frac{N}{D_1}\right), \frac{1}{2}\log\left(\frac{P'}{D_2}\right) \right\}.
\end{align*}
\end{proposition}  
\begin{proof}

\noindent \textit{Converse}

The converse follows from straightforward cutset bound arguments. 
The reader may notice that the expression given above is the maximum of the
Quadratic Gaussian Wyner-Ziv~\cite{Wyner} rate to decoder 1 and the Quadratic
Gaussian Wyner-Ziv rate to decoder 2, or equivalently the maximum of the two
cutset bounds. Clearly, this rate is the lowest possible for the given
distortions. 

\noindent \textit{Achievability}

We now show that it is also achievable using a modulo sum interpretation for \eqref{eqn:com_del}. Consider first encoding for decoder 1. From the Quadratic Gaussian Wyner-Ziv result, we know that side information at the encoder is redundant. Therefore, without loss of optimality, the encoder can code for decoder 1 using only $Y^n$, resulting in the codeword $U_Y^n$ and the corresponding index $I_Y$ after binning. Similarly, for decoder 2, the encoder can code for decoder 2 using $X^n$ only, resulting in the codeword $U_X^n$ and index $I_X$ after binning. The encoder then sends out the index $I_X \oplus I_Y$. Since decoder 1 has the $X^n$ sequence as side information, it knows the index $I_X$ and can therefore recover $I_Y$ from $I_X \oplus I_Y$. The same decoding scheme works as well for decoder 2. Therefore, we have shown the achievability of the given rate expression. We note further that this scheme corresponds to setting $U = (U_X, Y_Y)$ such that $U_X - X - Y - U_Y$ in rate expression (\ref{eqn:com_del}). 
\end{proof}

\textit{Remark 4.4}: As shown in our proof of achievability, if we have a good practical code for the Wyner-Ziv Quadratic Gaussian problem, then we also have a good practical code for the complementary delivery problem setting. We first develop two point to point codes: one for the Wyner-Ziv Quadratic Gaussian case with $X$ as the source and $Y$ as the side information, and another for the case where $Y$ is the source and $X$ is the side information. A good code for the complementary delivery setting is then obtained by taking the modulo sum of the indices produced by these two point to point codes.

We now turn to the doubly symmetric binary sources with Hamming distortion case. Here, the achievability scheme involves taking the modulo sum of the sources $X^n$ and $Y^n$.

\begin{proposition}
\textit{Doubly symmetric binary source with Hamming distortion.} Let $X \sim \Bern(1/2)$, $Y \sim \Bern(1/2)$, $X \oplus Y \sim \Bern(p)$ and both distortion measures be Hamming distortion. Assume, without loss of generality, that $D_1, D_2 \le p$. Then, 
\begin{align*}
R(D_1, D_2) = \max\{H(p) - H(D_1), H(p) - H(D_2)\}.
\end{align*}
\end{proposition}
\begin{proof}
The converse again follows from straightforward cutset bounds by considering decoders 1 and 2 individually. For the achievability scheme, let $Z = X\oplus Y$ and assume that $D_1 \le D_2$. Since $Z$ is i.i.d. $\Bern(p)$, using a point to point code for $Z$ at distortion $D_1$, we obtain a rate of $H(p) - H(D_1)$. Denote the reconstruction for $Z$ at time $i$ by $\Zh_i$. Decoder 1 reconstructs $Y_i$ by $\Yh_i = X_i \oplus \Zh_i$ for $i \in [1:n]$. Similarly, decoder 2 reconstructs $X$ by $\Xh_i = Y_i \oplus \Zh_i$ for $i \in [1:n]$. To verify that the distortion constraint holds, note that $d_1(Y_i, X_i \oplus \Zh_i) = Y_i \oplus X_i\oplus \Zh_i = Z_i \oplus \Zh_i$. Since $\Zh$ is a code that achieves distortion $D_1$, $\Yh$ satisfies the distortion constraint for decoder 1. The same analysis holds for decoder 2.
\end{proof}
\textit{Remark 4.5:} In this case, we only need a good code for the standard point to point rate distortion problem for a binary source. A good rate distortion code for a binary source is also a good code for the doubly symmetric binary source with Hamming distortion complementary delivery problem. 

\noindent \textit{Remark 4.6:} In our scheme, the reconstruction symbols at time
$i$ depend only on the received message and the side information at the decoder
at time $i$. Therefore, for this case, the rate distortion region for causal
reconstruction~\cite{Weissman--El-Gamal06} is the same as the rate distortion
region for noncausal reconstruction.

\section{Actions taken at the encoder} \label{sect:5}
We now turn to the case where the encoder takes action (figure~\ref{fig:3}) instead of the decoders. When the actions are taken at the encoder, the general rate-cost-distortion tradeoff region is open even for the case of a single decoder. Special cases which have been characterized includes the lossless case \cite{Permuter2010a}. In this section, we consider a special case of lossless source coding with $K$ decoders in which we can characterize the rate-cost tradeoff region.

\begin{theorem}~\label{thm:enc}
\textit{Special case of lossless source coding with actions taken at the encoder.} Let the action channel be given by the conditional distribution $P_{Y_1, Y_2, \ldots, Y_K|X,A}$. Assume further that $A = f_1(Y_1) = f_2(Y_2)= \ldots, f_K(Y_K)$. Then, the minimum rate required for lossless source coding with actions taken at the encoder and cost constraint $C$ is given by
\begin{align*}
R &= \min [\max_{j \in [1:K]} \left \{H(X|Y_j, A) \right\} -H(A|X)]^+,
\end{align*}
where minimization is over the joint distribution $p(x)p(a|x)p(y_1, y_2, \ldots,
y_K|x,a)$ such that $\E \Lambda(A) \le C$.
\end{theorem}
\begin{proof}

\noindent \textit{Converse}
The proof of converse is a straightforward extension from the single decoder case given in \cite{Permuter2010a}. We give the proof here for completeness. Consider the rate required for decoder $j$.
\begin{align*}
nR &\ge H(M) \\
&\ge H(M, X^n|Y_j^n) - H(X^n|M, Y_j^n) \\
& \ge H(M, X^n|Y_j^n) - n\e_n \\
& \stackrel{(a)}{=}  H(X^n|Y_j^n) - n\e_n \\
& \stackrel{(b)}{=} H(X^n) + H(Y_j^n|X^n, A^n) - H(Y_j^n) -n\e_n \\
&\ge \sum_{i=1}^n (H(X_i) + H(Y_{ji}| X_i, A_i) - H(Y_{j,i})) -n\e_n,
\end{align*}
where $(a)$ follows from the fact that $M$ is a function of $X^n$ and $(b)$ follows from $A^n$ being a function of $X^n$. The last step follows from $X^n$ being a discrete memoryless source; the action channel being memoryless and conditioning reduces entropy. As before, we define $Q$ to be an uniform random variable over $[1:n]$ independent of all other random variables to obtain
\begin{align*}
R &\ge H(X) + H(Y_{j}| X, A) - H(Y_{j}) - \e_n \\
& = H(X) + H(Y_{j}, X| A) - H(X|A) - H(Y_{j}) - \e_n \\
& = H( X| A, Y_{j}) + I(X;A) - I(Y_{j};A) - \e_n \\
& = H( X| A, Y_{j}) - H(A|X) - \e_n.
\end{align*}
The last step follows from the fact that $A = fj(Y_j)$. Combining the lower bounds over $K$ decoders then give us the achievable rate stated in the Theorem. 

\noindent \textit{Achievability}
We give a sketch of achievability since the techniques used are relatively straightforward. Assume first that $R >0$. We first bin the set of $X^n$ sequences to $2^{n(\max_{j\in [1:K]} H(Y_j|X,A) + \e)}$, $\Bc(M_X)$, $M_X \in [1:2^{n(\max_{j\in [1:K] H(Y_j}|X,A) + \e)}]$. Given an $x^n$ sequence, we first find the bin index $m_x$ such that $x^n \in \Bc(m_x)$. We then split $m_x$ into two sub-messages: $m_{xr} \in [1:2^{\max_{j \in [1:K]} \left \{H(X|Y_j, A) \right\} -H(A|X) + 2\e}]$ and $m_{xa} \in [1:2^{n(H(A|X) -\e)}]$. $m_{xr}$ is transmitted over the noiseless link, giving us the rate stated in the Theorem. As for $m_{xa}$, we will send the message through the action channel by treating the action channel as a channel with i.i.d. state $X$ noncausally known at the transmitter ($A$). We can therefore use Gel'fand Pinsker coding \cite{GP1980} for this channel.

Each decoder first decodes $m_{xa}$ from their side information $Y_j$. From the condition that $A = f_j(Y_j)$ for all $j$, we have $H(A|X) - \e = I(Y_j;A)- I(X;A) -\e$. From analysis of Gel'fand-Pinsker coding, since $|\mathcal{M}_{xa}| =I(Y_j;A)- I(X;A) -\e$, the probability of error in decoding $m_{xa}$ goes to zero as $n \to \infty$. The decoder then reconstructs $m_x$ from $m_{xr}$ and $m_{xa}$. It then finds the unique $\xh^n \in \Bc(m_x)$ that is jointly typical with $Y_j^n$ and $A^n$. Note that due to Gel'fand-Pinsker coding, the true $x^n$ sequence is jointly typical with $Y_j^n$ and $A^n$ with high probability. Therefore, the probability of error in this decoding step goes to zero as $n \to \infty$ since we have $2^{n(\max_{j\in [1:K]} H(Y_j|X,A) + \e)}$ bins. 

For the case where $R=0$, we send the entire message through the action channel. 
\end{proof}
\textit{Example 3:} Consider the case of $K = 2$ with switching dependent side information: $A = \{1,2\}$ and $(X,Y)\sim p(x,y)$ with $P_{Y_1, Y_2|X,A}$ specified by $Y_1 = Y, Y_2 = e$ when $A = 1$ and $Y_1 = e, Y_2 = Y$ when $A =2$. Note that $A$ is a function of $Y_1$, and also of $Y_2$. It therefore satisfies the condition in Theorem~\ref{thm:enc}. Let $\P(A = 1) = \alpha$, $\Lambda(A=1) =C_1$ and $\Lambda(A=2) = C_2$. The rate-cost tradeoff is characterized by 
\begin{align*}
R &= \max\{\alpha H(X|A=1, Y) + (1-\alpha)H(X|A=1), (1-\alpha) H(X|A=2, Y) + \alpha H(X|A=2)\} \\
& \qquad + H(X) - H_2(\alpha) -\alpha H(X|A=1) - (1-\alpha)H(X|A=2)
\end{align*} 
for some $p(a|x)$ satisfying $\alpha C_1 + (1-\alpha)C_2 \le C$. 
\section{Other settings} \label{sect:6}
In this section, we consider other settings involving multi-terminal source coding with action dependent side information. The first setting that we consider in this section generalizes \cite[Theorem 7]{Permuter2010a} to the case where there is a rate-limited link from the source encoder to the action encoder. The second setting we consider is a case of successive refinement with actions.

\subsection{Single decoder with Markov Form X-A-Y and rate limited link to action encoder}
In this subsection, we consider the setting illustrated in Figure \ref{fig:os1}.
Here, 
we have a single decoder with actions taken at an action encoder. The source
encoder have access to source $X^n$ and sends out two indices $M \in [1:2^{nR}]$
and $M_A \in [1:2^{nR_A}]$. The action encoder is a function $f: M_A \to A^n$.
In addition, we have the Markov relation $X-A-Y$. That is, the side information
$Y$ is dictated only by the action $A$ taken. The other
definitions remain the same and we omit them here.

\begin{figure}[!h]
\psfrag{x}[c]{$X^n$}
\psfrag{sc}[c]{Source Encoder}
\psfrag{ac}[c]{Action Encoder}
\psfrag{r}[c]{$R$}
\psfrag{ra}[c]{$R_A$}
\psfrag{y}[c]{$Y^n$}
\psfrag{dec}[c]{Decoder}
\psfrag{pyxa}[c]{$\P_{Y|A}$}
\psfrag{xh}[c]{$\Xh^n$}
\psfrag{a}[c]{$A^n$}
\begin{center}
\scalebox{0.85}{\includegraphics{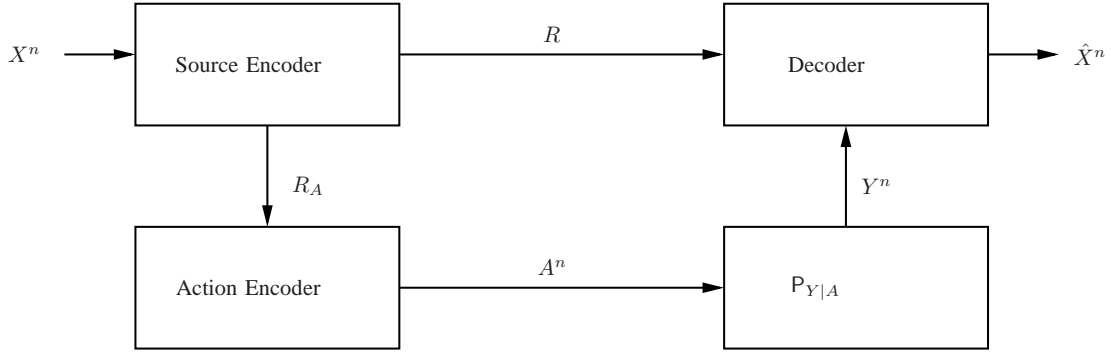}} 
\end{center}
\caption{Lossy source coding with rate limited link to action encoder} \label{fig:os1}
\vspace{-5pt}
\end{figure}
 
\begin{proposition}\label{prop:rlimit}
$R(D,C)$ for the setting shown in figure \ref{fig:os1} is given by
\begin{align*}
R(D,C) = \min \max\{I(X;\Xh) -R_A, I(X;\Xh) - I(A;Y)\},
\end{align*}
where the minimization is over $p(x)p(a)p(y|a)p(\xh|x)$ satisfying the cost and
distortion constraints $\E d(X,\Xh) \le D$ and $\E \Lambda(A) \le C$. 
\end{proposition}

\noindent \textit{Remark 6.1:} If we set $R_A = \infty$ in Proposition \ref{prop:rlimit}, then we recover the result in \cite[Theorem 7]{Permuter2010a}. Essentially, the source encoder tries to send as much information as possible through the rate limited action link until the link saturates.
\begin{proof}

\noindent \textit{Achievability:} The achievability is straightforward. Using standard rate distortion coding, we cover $X^n$ with $2^{n(I(X;\Xh)+\e)}$ $\Xh^n$ codewords. Given a source sequence $x^n$, we find an $\Xh^n$ that is jointly typical with $x^n$. We then split the index $M_X$ corresponding to the chosen $\Xh^n$ codeword into two parts: $M_A \in [1:2^{n(\min\{R_A, I(A;Y)\} +\e)}]$ and $M \in [1:2^{nR}]$. The action encoder takes the index and transmit it through the action channel. Since the rate of $M_A$ is less than $I(A;Y) -\e$, the decoder can decode $M_A$ with high probability of success. It then combines $M_A$ with $M$ to obtain the index of the reconstruction codeword $\Xh^n$.

\noindent \textit{Converse}
Given a code that satisfy the distortion and cost constraints, we have
\begin{align*}
nR &\ge H(M) \\
&= I(X^n;M) \\
& \ge I(X^n;M) - I(X^n;Y^n)\\
& = I(X^n;\Xh^n) - I(X^n, M_A;Y^n) \\
& \stackrel{(a)}{\ge} \sum_{i=1}^n I(X_i;\Xh_i) - I(X^n, M_A, A^n;Y^n) \\
& \stackrel{(b)}{\ge} \sum_{i=1}^n I(X_i;\Xh_i) - I( M_A, A^n;Y^n).
\end{align*}
$(a)$ follows from the fact that $A^n$ is a function of $M_A$. $(b)$ follows from the Markov chain $X-A-Y$. Now, it is easy to see that $I(M_A, A^n;Y_i) \le \min\{nR_A, \sum_{i=1}^n I(A_i;Y_i)\}$. The bound on the rate is then single letterized in the usual manner, giving us
\begin{align*}
R(D,C) = \min \max\{I(X;\Xh) -R_A, I(X;\Xh) - I(A;Y)\},
\end{align*}
for some $p(a,\xh|x)$ satisfying the distortion and cost constraints. Finally, we note that $p(a, \xh|x)$ can be restricted to the form $p(a)p(\xh|x)$. To see this, note that none of the terms depend on the joint $p(a,\xh|x)$. Furthermore, due to the Markov conditon $X-A-Y$, it suffices to consider $A$ independent of $X$, giving us the p.m.f in the Proposition.
\end{proof}

\subsection{Successive refinement with actions}
The next setup that we consider is a case of successive refinement \cite{Equitz1993}, \cite{Rimoldi1994} with actions taken at the ``more capable'' decoder. The setting is shown in Figure \ref{fig:os2}.

\begin{figure}[!h]
\psfrag{x}[c]{$X^n$}
\psfrag{sc}[c]{Source Encoder}
\psfrag{ac}[c]{Action Encoder}
\psfrag{dec2}[c]{Decoder 2}
\psfrag{r}[c]{$R_1$}
\psfrag{y}[c]{$Y^n$}
\psfrag{dec}[c]{Decoder 1}
\psfrag{pyxa}[c]{$\P_{Y|X,A}$}
\psfrag{xh}[l]{$\Xh^n_1, D_1$}
\psfrag{xh2}[l]{$\Xh^n_2, D_2$}
\psfrag{a}[c]{$A^n$}
\psfrag{r1}[c]{$R_2$}
\begin{center}
\scalebox{0.85}{\includegraphics{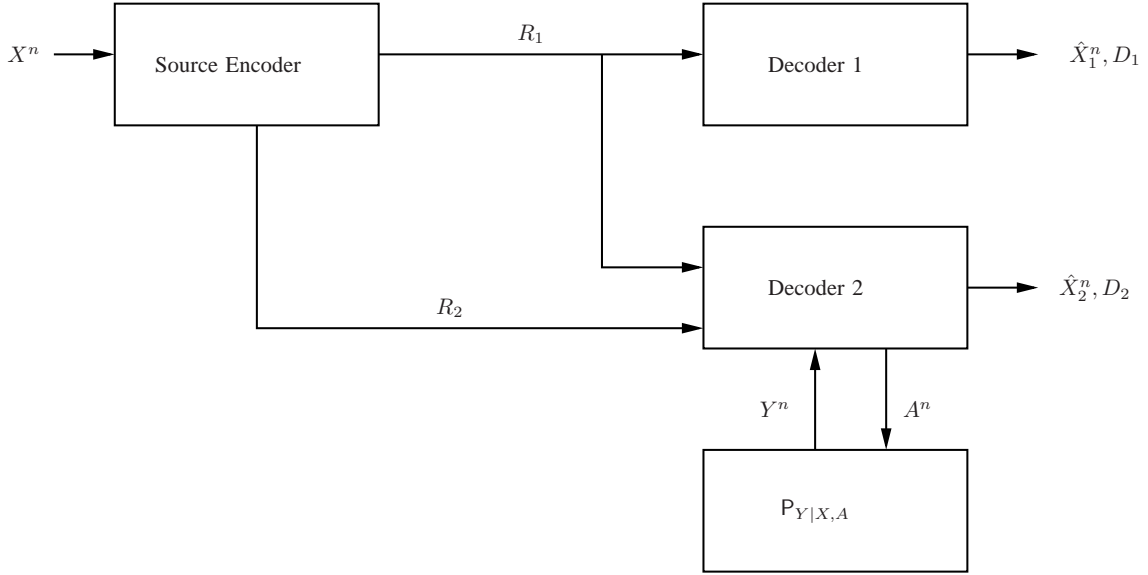}} 
\end{center}
\caption{Successive refinement with actions} \label{fig:os2}
\vspace{-5pt}
\end{figure}
 
\begin{proposition} \label{prop:sr}
\textit{Successive refinement with actions taken at the more capable decoder}
For the setting shown in figure \ref{fig:os2}, the rate distortion cost tradeoff region is given by
\begin{align*}
R_1 &\ge I(X;\Xh_1), \\
R_1+ R_2 &\ge I(X;\Xh_1, A) + I(X;U|\Xh_1, Y, A) 
\end{align*}
for some $p(\xh_1,a,u|x)$ satisfying
\begin{align*}
\E d_1(X, \Xh_1) \le D_1, \\
\E d_1(X, \Xh_2(U, Y, A)) \le D_2, \\
\E \Lambda(A) \le C.
\end{align*}
\end{proposition}
The cardinality of the auxiliary $U$ may be upper bounded by $|\Uc| \le |\Xc||\mathcal{\Xh}_1||\mathcal{A}|+1$. 

If we restrict $R_2 = 0$, then Proposition \ref{prop:sr} gives the rate-distortion-cost tradeoff region for a special case of Proposition \ref{prop:HB}. That is, the case when $Y_2 = \emptyset$ and actions are taken only at decoder 1. 
\begin{proof}

\noindent \textit{Achievability:} We give the case where $R_1 = I(X;\Xh_1) + \e$ and $R_2 = I(X;A|\Xh_1) + I(X;U|\Xh_1, Y, A) + 3\e $. The general region stated in the Proposition can then be obtained by rate splitting of $R_2$.
\subsection*{Codebook generation}
\begin{itemize}
\item Generate $2^{nR_1}$ $\Xh_1^n(m_1)$ sequences according to $\prod_{i=1}^n p(\xh_{1i})$, $m_1 \in [1:2^{nR}]$.
\item For each $\Xh_1^n(m_1)$ sequence, generate $2^{n(I(X;A|\Xh_1)+\e)}$ $A^n(m_1, m_{21})$, sequences according to $\prod_{i=1}^n p(a_i|\xh_{1i})$.
\item For each $\Xh_1^n(m_1)$ and $A^n(m_1, m_2)$ sequence pair, generate $2^{n(I(X;U|\Xh_1,A)+\e)}$ $U^n(m_1, m_{21}, l_{22})$, sequences according to $\prod_{i=1}^n p(u_i|\xh_{1i}, a_i)$.
\item Partition the set of $l_{22}$ indices into $2^{I(X;U|\Xh_1, Y, A) + 2\e}$ bins, $\Bc (m_{1}, m_{21},m_{22})$, $m_{22} \in [1:2^{n(I(X;U|\Xh_1, Y, A)+2\e)}]$.
\end{itemize}
\subsection*{Encoding}
\begin{itemize}
\item Given a sequence $x^n$, the encoder first looks for an $\xh_1^n(m_1)$ sequence such that $(x^n, \xh_1^n) \in \aep$. This step succeeds with high probability since $R_1 =  I(X;\Xh_1) + \e$.
\item Next, the encoder looks for an $A^n(m_1, m_{21})$ sequence such that $(x^n, a^n, \xh_1^n) \in aep$. This step succeeds with high probability since we have $2^{n(I(X;A|\Xh_1)+\e)}$ $A^n$ sequences. 
\item The encoder then looks for an $U^n(m_1, m_{21}, l_{22})$ sequence such that $(x^n, a^n, \xh_1^n, u^n) \in aep$. This step succeeds with high probability since we have $2^{n(I(X;U|\Xh_1,A)+\e)}$ $U^n$ sequences. 
\item It then finds the bin index such that $l_{22} \in \Bc(m_1, m_{21}, m_{22})$. 
\item The encoder sends out the indices $m_1$ over the link $R_1$ and $m_{21}$ and $m_{22}$ over the link $R_2$, giving us the stated rates.
\end{itemize}
\subsection*{Decoding and reconstruction}
\begin{itemize}
\item Since decoder 1 has index $m_1$, it reconstructs $x^n$ using $\xh_1(m_1)^n$. Since $(x^n, \xh_1^n)$ are jointly typical with high probability, the expected distortion satisfies the $D_1$ distortion constraint to within $\e$.
\item For decoder 2, from $m_1$ and $m_{21}$, it recovers the action sequence $a^n(m_1, m_{21})$. It then takes the action $a^n(m_1, m_{21})$ to obtain it's side information $Y^n$. With the side information, it recovers the $u^n$ sequence by looking for the unique $\lh_{22} \in  \Bc(m_1, m_{21}, m_{22})$ such that $(u^n(m_{1}, m_{21}, \lh_{22}), \xh_1^n, a^n, Y^n) \in \aep$. Since there are only $2^{n(I(U;Y|\Xh_1, A) -\e)}$ $U^n$ sequences in the bin and $(u^n(m_{1}, m_{21}, l_{22}), \xh_1^n, a^n, Y^n) \in \aep$ with high probability from the fact that $Y$ is generated i.i.d. according to $p(y|a_i,x_i)$, the probability of error goes to zero as $n \to \infty$. Decoder 2 then reconstructs $x^n$ using $\xh_{2i}(a_i, u_i, y_i)$ for $i\in [1:n]$.
\end{itemize}

\noindent \textit{Converse:} We consider only the lower bound for $R_1 + R_2$. The lower bound for $R_1$ is straightforward.
Given a code which satisfies the distortion and cost constraints, we have
\begin{align*}
n(R_1 + R_2) &\ge H(M_1, M_2) \\
& = H(M_1, M_2, A^n, \Xh^n_1) \\
&= H(A^n, \Xh_1^n) +H(M_1, M_2|A^n, \Xh_1^n) \\
& \ge I(X^n;A^n, \Xh_1^n) + H(M_1, M_2|A^n, \Xh_1^n, Y^n) - H(M_1, M_2|Y^n, A^n, \Xh_1^n, X^n)\\
& = I(X^n;A^n, \Xh_1^n) + I(X^n; M_1, M_2|A^n, \Xh_1^n, Y^n) \\
& = I(X^n;A^n, \Xh_1^n) + H(X^n|A^n, \Xh_1^n, Y^n)- H( X^n|A^n, \Xh_1^n, Y^n, M_1,M_2) \\
& = I(X^n;A^n, \Xh_1^n) + H(X^n, Y^n|A^n, \Xh_1^n) - H(Y^n|\Xh_1^n, A^n)- H( X^n|A^n, \Xh_1^n, Y^n, M_1,M_2) \\
& = H(X^n) - H(X^n|A^n, \Xh_1^n) + H(X^n, Y^n|A^n, \Xh_1^n) - H(Y^n|\Xh_1^n, A^n)- H( X^n|A^n, \Xh_1^n, Y^n, M_1,M_2) \\
& = H(X^n) + H(Y^n|X^n, A^n, \Xh_1^n) - H(Y^n|\Xh_1^n, A^n)- H( X^n|A^n, \Xh_1^n, Y^n, M_1,M_2) \\
& \ge \sum_{i=1}^n(H(X_i) + H(Y_i|X^n, A^n, \Xh_1^n, Y^{i-1}) - H(Y_i|\Xh_1^n, A^n, Y^{i-1})- H( X_i|A^n, \Xh_1^n, Y^n, M_1,M_2)) \\
& \stackrel{(a)}{\ge} \sum_{i=1}^n(H(X_i) + H(Y_i|X_i, A_i, \Xh_{1i}) - H(Y_i|\Xh_{1i}, A_i)- H( X_i|A^n, \Xh_1^n, Y^n, M_1,M_2)) \\
& = \sum_{i=1}^n(H(X_i) + H(Y_i|X_i, A_i, \Xh_{1i}) - H(Y_i|\Xh_{1i}, A_i)- H( X_i|A^n, \Xh_1^n, Y^n, M_1,M_2)) \\
& \ge \sum_{i=1}^n(H(X_i) + H(Y_i|X_i, A_i, \Xh_{1i}) - H(Y_i|\Xh_{1i}, A_i)- H( X_i|U_i, A_i, Y_i, \Xh_{1i})) \\
\end{align*}
$(a)$ follows from the Markov Chain $(X^{n \backslash i}, A^{n\backslash i }, \Xh_1^{n\backslash i}, Y^{i-1}) - (\Xh_i,X_i,A_i) - Y_i$ and the last step follows from defining $U_i = (M_1, M_2, Y^{n \backslash i}, A^{n\backslash i})$. The proof is then completed in the usual manner by defining the time sharing uniform random variable $Q$ and $U = (U_Q, Q)$, giving us
\begin{align*}
R_1 + R_2  &\ge H(X) + H(Y|X, A, \Xh_1) - H(Y|\Xh_1, A) - H(X|U, A, Y, \Xh_1) \\
& = I(X;\Xh_1, A) + I(X;U|\Xh_1, Y, A).
\end{align*}
The fact that $\Xh_2$ is a function of $U$, $Y$ and $A$, which is straightforward. Finally, the cardinality bound on $U$ may be obtained from standard techniques. Note that we need $|\mathcal{\Xh}_1||\Xc||\mathcal{A}|-1$ letters to preserve $p(u,a,x)$ and two more to preserve the rate and distortion constraints.
\end{proof}

\noindent \textit{Remark 6.2}: An interesting question to explore characterizing the more general case when degraded side information is also available at decoder 1. That is, we have the side informations $Y_1$ at decoder 1 and $Y_2$ at decoder 2 are generated by a discrete memoryless channel $\P_{Y_1, Y_2|X,A}$ such that $(X,A) - (Y_2,A) - (Y_1,A)$. This generalized setup would allow us to generalize Proposition \ref{prop:HB} entirely and also leads to a generalization of successive refinement for the Wyner-Ziv problem in \cite{Steinberg--Merhav2004} to the action setting.   
\section{Conclusion} \label{sect:7}
In this paper, we considered an important class of multi-terminal source coding problems, where the encoder sends the description of the source to the decoders, which then take cost-constrained actions that affect the quality or availability of side information.  We computed the optimum rate region for lossless compression, while for the lossy case we provide a general achievability scheme that is shown to be optimal for a number of special cases, one of them being the generalization of \textit{Heegard-Berger-Kaspi} setting. (cf. \cite{Heegard--Berger1985}, \cite{Kaspi1994}). In all these cases in addition to a standard achievability argument, we also provided a simple scheme which has a \textit{modulo sum }interpretation. The problem where the encoder takes actions rather than the decoders, was also considered. Finally, we extended the scope to additional multi-terminal source coding problems such as successive refinement with actions. 

\section*{Acknowledgment}
We thank Professor Haim Permuter and Professor Yossef Steinberg for helpful discussions. The authors' research was partially supported by NSF grant CCF-0729195 and the Center for Science of Information (CSoI), an NSF Science and Technology Center, under grant agreement CCF-0939370; the Scott A. and Geraldine D. Macomber Stanford Graduate Fellowship and the Office of Technology Licensing Stanford Graduate Fellowship. 
\bibliographystyle{IEEEtran}
\bibliography{bibfile}
\appendices
\section{Achievability Sketch for Theorem \ref{thm:2}} \label{appen1}
\subsection*{Codebook generation}
\begin{itemize}
\item Generate $2^{n(I(X;A) + \e)}$ $A^n(l_a)$, $l_a \in [1:2^{n(I(X;A) + \e)}]$, sequences according to $\prod_{i=1}^n p(a_i)$. 
\item For each $A^n$ sequence, generate $2^{n(I(U;X|A) + \e)}$ $U^n(l_a,l_0)$, $l_0 \in [1:2^{n(I(U;X|A) + \e)}]$, sequences according to $\prod_{i=1}^n p_{U|A}(u_i|a_i)$.
\item Partition the set of indices corresponding to the $U^n$ codewords uniformly to $2^{n(\max\left\{I(X;U|A, Y_1), I(X;U|A,Y_2)\right\} + 2\e)}$ bins, $\Bc_{U}(l_a, m_0)$, $m_0 \in [1:2^{n(\max\left\{I(X;U|A, Y_1), I(X;U|A,Y_2)\right\} + 2\e)}]$.
\item For each pair of $A^n$ and $U^n$ sequences, generate $2^{n(I(V_1;X|A,U) + \e)}$ $V_1^n(l_a, l_0, l_1)$, $l_1 \in [1: 2^{n(I(V_1;X|A,U) + \e)}]$, sequences according to $\prod_{i=1}^n p_{V_1|A,U}(v_{1i}|a_i, u_i)$.
\item Partition the set of indices corresponding to the $V^n_1$ codewords uniformly to $2^{n( I(X;V_1|U, A, Y_1) + 2\e)}$ bins, $\Bc_{V_1}(l_a, l_0, m_1)$, $m_1 \in [1:2^{n(I(X;V_1|U, A, Y_1) + 2\e)}]$.
\item For each pair of $A^n$ and $U^n$ sequences, generate $2^{n(I(V_2;X|A,U) + \e)}$ $V_2^n(l_a, l_0, l_2)$, $l_1 \in [1: 2^{n(I(V_2;X|A,U) + \e)}]$, sequences according to $\prod_{i=1}^n p_{V_2|A,U}(v_{2i}|a_i, u_i)$.
\item Partition the set of indices corresponding to the $V^n_2$ codewords uniformly to $2^{n( I(X;V_2|U, A, Y_2) + 2\e)}$ bins, $\Bc_{V_2}(l_a, l_0, m_2)$, $m_2 \in [1:2^{n(I(X;V_2|U, A, Y_2) + 2\e)}]$.
\end{itemize}
\subsection*{Encoding}
\begin{itemize}
\item Given an $x^n$ sequence, the encoder first looks for an $a^n(l_a)$ sequence such that $(x^n, a^n) \in \aep$.  If there is none, it outputs and index chosen uniformly at random from the set of possible $l_a$ indices. If there is more than one, it outputs an index chosen uniformly at random from the set of feasible indices. Since there are $2^{n(I(X;A) + \e)}$ such sequences, the probability of error $\to 0$ as $n \to \infty$.
\item The encoder then looks for a $u^n(l_a, l_0)$ sequence that is jointly typical with $(a^n(l_a), x^n)$. If there is none, it outputs and index chosen uniformly at random from the set of possible $l_0$ indices. If there is more than one, it outputs an index chosen uniformly at random from the set of feasible indices. Since there are $2^{n(I(U;X|A) + \e)}$ such sequences, the probability of error $\to 0$ as $n \to \infty$.
\item Next, the encoder looks for a $v_1^n(l_a, l_0, l_1)$ sequence that is jointly typical with $(a^n(l_a), u^n(l_0),  x^n)$. If there is none, it outputs and index chosen uniformly at random from the set of possible $l_0$ indices. If there is more than one, it outputs an index chosen uniformly at random from the set of feasible indices. Since there are $2^{n(I(V_1;X|A,U) + \e)}$ such sequences, the probability of error $\to 0$ as $n \to \infty$.
\item Next, the encoder looks for a $v_2^n(l_a, l_0, l_2)$ sequence that is jointly typical with $(a^n(l_a), u^n(l_0), x^n)$. If there is none, it outputs and index chosen uniformly at random from the set of possible $l_0$ indices. If there is more than one, it outputs an index chosen uniformly at random from the set of feasible indices. Since there are $2^{n(I(V_2;X|A,U) + \e)}$ such sequences, the probability of error $\to 0$ as $n \to \infty$.
\item The encoder then sends out the indices $l_a$, $m_0$, $m_1$ and $m_2$ such that $l_0 \in \Bc_U(l_a, m_0)$, $l_1 \in \Bc_{V_1}(l_a, l_0, m_1)$ and $l_2 \in \Bc_{V_2}(l_a, l_0, m_2)$.
\end{itemize}
\subsection*{Decoding and reconstruction}
Decoder 1:
\begin{itemize}
\item Decoder 1 first takes the action sequence $a^n(l_a)$ to obtain the side information $Y_1^n$. We note that if \\$(a^n(l_a), x^n, u^n(l_a, l_0), v_1^n(l_a, l_0, l_1)) \in \aep$, then $\P\{(a^n(l_a), x^n, u^n(l_a, l_0), v_1^n(l_a, l_0, l_1), Y_1^n) \in \aep\} \to 1$ as $n \to \infty$ by the conditional typicality lemma \cite[Chapter 2]{El-Gamal--Kim2010} and the fact that $Y_1^n \sim \prod_{i=1}^n p(y_{1i}|x_i, a_i)$.
\item Decoder 1 then decodes $U^n$. it does this by finding the unique $\lh_0$ such that $u^n(l_a, \lh_0) \in \Bc_U(l_a, m_0)$. If there is none or more than one such $\lh_0$, an error is declared. Following standard analysis for the Wyner-Ziv setup (see for e.g. \cite[Chapter 12]{El-Gamal--Kim2010}), the probability of error goes to zero as $n \to \infty$ since there are less than or equal to $2^{n(I(U;Y_1|A) - \e)}$ $U^n$ sequences within each bin. 
\item Similarly, decoder 1 decodes $V^n_1$. It does this by finding the unique $\lh_1$ such that $v_1^n(l_a, \lh_0, \lh_1) \in \Bc_{V_1}(l_a, \lh_0, m_1)$. If there is none or more than one such $\lh_1$, an error is declared. As with the previous step, the probability of error goes to zero as $n \to \infty$ since there are only $2^{n(I(V_1;Y_1|A, U) - \e)}$ $V_1^n$ sequences within each bin. 
\item Decoder 1 then reconstructs $x^n$ as $\xh_{1i}(a_i(l_a), u_i(l_a, \lh_0), v_{1i}(l_a, \lh_0, \lh_1), y_{1i})$ for $i \in [1:n]$.
\end{itemize}

Decoder 2: As the decoding steps for decoder 2 are similar to that for 1, we will only mention the differences here. That is, decoder 2 uses side information $Y_2^n$ instead of $Y_1^n$ to perform the decoding operations and instead of decoding $V_1^n$, decoder 2 decodes $V_2^n$.
\begin{itemize}
\item Decoder 2 decodes $V^n_2$. It does this by finding the unique $\lh_2$ such that $v_2^n(l_a, \lh_0, \lh_2) \in \Bc_{V_2}(l_a, \lh_0, m_2)$. If there is none or more than one such $\lh_2$, an error is declared. As with the previous step, the probability of error goes to zero as $n \to \infty$ since there are only $2^{n(I(V_2;Y_2|A, U) - \e)}$ $V_2^n$ sequences within each bin. 
\item Decoder 1 then reconstructs $x^n$ as $\xh_{2i}(a_i(l_a), u_i(l_a, \lh_0), v_{2i}(l_a, \lh_0, \lh_2), y_{2i})$ for $i \in [1:n]$.
\end{itemize}

\subsection*{Distortion and cost constraints}
\begin{itemize}
\item For the cost constraint, since the chosen $A^n$ sequence is typical with high probability, $\E\Lambda(A^n) \le C + \e$ by the typical average lemma \cite[Chapter 2]{El-Gamal--Kim2010}.
\item For the distortion constraints, since the probability of ``error'' goes to zero as $n \to \infty$ and we are dealing only with finite cardinality random variables, following the analysis in \cite[Chapter 3]{El-Gamal--Kim2010}, we have
\begin{align*}
\frac{1}{n}\E d_1(X^n, \Xh_1^n) &\le D_1 + \e, \\
\frac{1}{n}\E d_2(X^n, \Xh_2^n) &\le D_2 + \e. 
\end{align*}
\end{itemize}

\end{document}